\documentclass[a4paper,11pt]{article}
\usepackage{amsmath,amssymb,amsthm,amsfonts}
\usepackage{graphicx,caption,subcaption,fullpage}
\usepackage[pdftex,colorlinks=true,linkcolor=blue,citecolor=blue,urlcolor=black]{hyperref}

\usepackage{tikz}
\usetikzlibrary{positioning,calc,decorations.pathreplacing}

\def\<{\langle}
\def\>{\rangle}
\def\Re{\text{Re}}
\newcommand{\C}{\mathbb{C}}
\newcommand{\R}{\mathbb{R}}

\DeclareMathOperator{\poly}{poly}
\newcommand{\sham}{{\sc $\mathcal{S}$-Hamiltonian}}
\newcommand{\spham}{{\sc $\mathcal{S}^+$-Hamiltonian}}
\newcommand{\lham}{{\sc Local Hamiltonian}}

\makeatletter
\newtheorem*{rep@theorem}{\rep@title}
\newcommand{\newreptheorem}[2]{%
\newenvironment{rep#1}[1]{%
 \def\rep@title{#2 \ref{##1} (restated)}%
 \begin{rep@theorem}}%
 {\end{rep@theorem}}}
\makeatother

\newtheorem{theorem}{Theorem}
\newtheorem{lemma}{Lemma}

\newtheorem{definition}{Definition}
\newreptheorem{theorem}{Theorem}
\newreptheorem{lemma}{Lemma}

\tikzstyle{qubit}=[circle,draw,fill,thick,
inner sep=0pt,minimum size=2mm]
\tikzstyle{mqubit}=[circle,draw,fill=white,thick,
inner sep=0pt,minimum size=2mm]
\tikzstyle{heavy}=[ultra thick]

\begin{document}
\title{The complexity of antiferromagnetic interactions and 2D lattices}
\author{Stephen Piddock\footnote{\url{stephen.piddock@bristol.ac.uk}}\ \ and Ashley Montanaro\\[8pt]
{\small Department of Computer Science, University of Bristol, UK}
}

\maketitle
\begin{abstract}
Estimation of the minimum eigenvalue of a quantum Hamiltonian can be formalised as the Local Hamiltonian problem. We study the natural special case of the Local Hamiltonian problem where the same 2-local interaction, with differing weights, is applied across each pair of qubits. First we consider antiferromagnetic/ferromagnetic interactions, where the weights of the terms in the Hamiltonian are restricted to all be of the same sign. We show that for symmetric 2-local interactions with no 1-local part, the problem is either QMA-complete or in StoqMA. In particular the antiferromagnetic Heisenberg and antiferromagnetic XY interactions are shown to be QMA-complete. We also prove StoqMA-completeness of the antiferromagnetic transverse field Ising model. Second, we study the Local Hamiltonian problem under the restriction that the interaction terms can only be chosen to lie on a particular graph. We prove that nearly all of the QMA-complete 2-local interactions remain QMA-complete when restricted to a 2D square lattice. Finally we consider both restrictions at the same time and discover that, with the exception of the antiferromagnetic Heisenberg interaction, all of the interactions which are QMA-complete with positive coefficients remain QMA-complete when restricted to a 2D triangular lattice.
\end{abstract}

\section{Introduction}
Calculation of the ground-state energy of quantum Hamiltonians with interactions obeying locality constraints is a fundamental problem in physics. This is encapsulated within quantum information theory as the \lham\ problem. A Hamiltonian $H$ on $n$ qubits (a Hermitian matrix acting on $(\C^2)^{\otimes n}$) is said to be $k$-local if it can be written as $H = \sum_j H_j$, where each interaction term $H_j$ acts non-trivially on at most $k$ qubits. Let $\lambda(H)$ be the ground-state energy (the lowest eigenvalue) of $H$, and assume that $\|H_j\| \leqslant \poly(n)$ for all $j$. Then the \lham\ problem is to determine whether $\lambda(H) \leqslant a$, or $\lambda(H) \geqslant b$, for some $a$, $b$ such that $b-a \geqslant 1/\poly(n)$.

It was shown by Kitaev~\cite{Kitaev2002} that, even for $k=5$, \lham\ is QMA-complete. QMA is a complexity class which is the quantum analogue of the classical complexity class NP~\cite{papadimitriou94}. A problem is said to be ``complete'' for a complexity class if it is contained within that class, and any other problem in that class can efficiently be reduced to it; if a problem is QMA-complete, it is likely that there is no efficient quantum (or classical) algorithm to solve it. Kempe, Kitaev and Regev~\cite{Kempe2004} later showed that the \lham\ problem remains QMA-complete even for $k=2$, whereas the special case where $k=1$ can easily be solved efficiently. The Hamiltonians occurring in this hardness proof are somewhat artificial. A number of subsequent works have therefore attempted to show that {\sc Local Hamiltonian} remains QMA-complete, even with more physically realistic restrictions on the interactions~\cite{Oliveira2005,biamonte08,Schuch2009,schuch11,cubitt14,childs14}. This simultaneously makes the theory more useful for applications and provides insight into the crucial features of physical theories that determine their complexity. For example, it allows one to ask, and hopefully answer, questions of the form ``is the Heisenberg model on a square lattice more complex than the Ising model on a triangular lattice?''.

A particularly natural restriction to consider is to restrict the types of interactions allowed. Let $\mathcal{S}$ be a set of allowed interaction terms on at most 2 qubits. The \sham\ problem is a restriction of \lham\ where $k=2$ and the Hamiltonian $H=\sum \alpha_i H_i$ can be written as a linear combination of terms $H_i \in \mathcal{S}$, and where each $\alpha_i$ is a positive or negative real weight such that $|\alpha_i| \leqslant \poly(n)$. For example, the general Ising model, where $H = \sum_{i,j} \alpha_{ij} Z_i Z_j$, corresponds to $\mathcal{S} = \{ZZ\}$. Similarly, the (general) Heisenberg model is $\{XX+YY+ZZ\}$; the XY model is $\{XX+YY\}$; the Ising model with transverse magnetic field is $\{ZZ,X\}$. Note that $X$, $Y$, $Z$ are the Pauli matrices and we omit tensor product symbols for readability.

A complexity classification of the \sham\ problem was obtained in~\cite{cubitt14}. This can be stated as follows:

\begin{theorem}[Cubitt and Montanaro~\cite{cubitt14}, Bravyi and Hastings~\cite{Bravyi2014a}]
\label{thm:sham}
Let $\mathcal{S}$ be an arbitrary fixed subset of Hermitian matrices on at most 2 qubits. Then:
\begin{itemize}
\item If every matrix in $\mathcal{S}$ is 1-local, \sham\ is in P;
\item Otherwise, if there exists $U \in SU(2)$ such that $U$ diagonalises all 1-qubit matrices in $\mathcal{S}$, and $U^{\otimes 2}$ diagonalises all 2-qubit matrices in $\mathcal{S}$, then \sham\ is NP-complete;
\item Otherwise, if there exists $U \in SU(2)$ such that, for each 2-qubit matrix $H_i \in \mathcal{S}$, $U^{\otimes 2} H_i (U^{\dag})^{\otimes 2} = \alpha_i Z^{\otimes 2} + A_iI + IB_i$, where $\alpha_i \in \R$ and $A_i$, $B_i$ are arbitrary single-qubit Hermitian matrices, then \sham\ is StoqMA-complete;
\item Otherwise, \sham\ is QMA-complete.
\end{itemize}
\end{theorem}

The third case in Theorem \ref{thm:sham} was originally shown to be complete for a complexity class corresponding to the transverse Ising model in~\cite{cubitt14}; this was later sharpened to StoqMA-completeness by Bravyi and Hastings~\cite{Bravyi2014a}. The \lham\ problem can be seen as a quantum generalisation of classical constraint satisfaction problems, with each term in the Hamiltonian corresponding to a constraint. From this perspective, Theorem \ref{thm:sham} is a quantum analogue of a classical dichotomy theorem of Schaefer~\cite{schaefer78} classifying the complexity of constraint satisfaction problems in terms of the types of allowed constraints.

Each of the above complexity classes (P, NP, StoqMA, QMA) corresponds to a computational model. P is polynomial-time classical computation (decision problems that can be solved efficiently by a classical computer). NP (``nondeterministic polynomial-time'') corresponds to decision problems whose solutions can be checked efficiently. StoqMA is a complexity class introduced by Bravyi, Bessen and Terhal~\cite{bravyi06a} whose definition as a computational model is somewhat technical, but which corresponds to the special case of the \lham\ problem where the matrices $H_j$ have non-positive off-diagonal elements (known as ``stoquastic''). Finally, QMA is the complexity class corresponding to decision problems whose solutions can be checked in polynomial time by a {\em quantum} computer. We have P $\subseteq$ NP $\subseteq$ StoqMA $\subseteq$ QMA and it is believed that each of these inclusions is strict. Some natural examples of interactions falling into each of these classes: the Ising model is NP-complete~\cite{barahona82}; the transverse Ising model is StoqMA-complete~\cite{Bravyi2014a}; and the Heisenberg and XY models are QMA-complete~\cite{cubitt14}.

Although Theorem~\ref{thm:sham} gives a precise and complete classification of the complexity of 2-qubit interactions, it suffers from several shortcomings which make it less meaningful as a statement about physics:
\begin{enumerate}
\item The interactions do not have any constraints on their spatial locality. They can occur across large distances and each individual qubit can interact with arbitrarily many others;
\item The interaction terms $H_j$ can appear with either positive or negative signs in the final Hamiltonian $H$, corresponding to two quite different physical interpretations;
\item The weights of individual interactions can be as large as $\poly(n)$, and the precision with which we are asked to estimate the ground-state energy scales as $O(1/\poly(n))$. In physical systems the interactions are usually of weight $O(1)$.
\end{enumerate}
In this work we address the first and second of these issues. We do not address the third (but see~\cite{cao14,childs14,Childs2015} for some cases where QMA-completeness has been proven with interactions of strength $O(1)$).

\subsection{Statement of results}

We study the \sham\ problem with physically motivated restrictions. The first restriction is the additional constraint that all weights $\alpha_i$ must satisfy $\alpha_i \geqslant 0$. We will label this problem \spham, although for the majority of this paper we consider sets of only one element. Physically, this restriction allows us to characterise the complexity of models which have only either antiferromagnetic or ferromagnetic interactions. This restriction is also useful when viewing the \sham\ problem as the quantum analogue of a weighted constraint satisfaction problem. Indeed, from this point of view it is difficult to interpret the meaning of negative coefficients.

Our main result in this setting is the following theorem:

\begin{theorem}
\label{thm:positive}
Given a 2-qubit interaction $H = \alpha XX+\beta YY + \gamma ZZ$, the problem $\{H\}^+$\textsc{-Hamiltonian} is either:

\emph{\textbf{i) }} QMA-complete, if $\alpha+\beta>0,$ $\alpha+\gamma>0$ and $\beta+\gamma>0$;

\emph{\textbf{ii)}} in StoqMA, otherwise.
\\Furthermore, if  $\alpha=-\beta\neq 0$, $\alpha+\gamma>0$ and $\beta+\gamma>0$, then $\{H\}^+$\textsc{-Hamiltonian} is StoqMA-complete. If $\alpha=\beta$ and $\gamma \leqslant -|\alpha|$, $\{H\}^+$\textsc{-Hamiltonian} is in P.
\end{theorem}
The latter part of the theorem can easily be seen to hold under relabelling $\alpha$, $\beta$, $\gamma$. It was shown in \cite{cubitt14} that if at least two of $\alpha,\beta,\gamma$ are non-zero, then the corresponding problem $\{H\}$\textsc{-hamiltonian} is QMA-complete; and so Theorem \ref{thm:positive} shows that the restriction of having positive coefficients does indeed make a significant difference to the complexity.

Observe that we can transform any 2-qubit interaction $H$ which is symmetric under swapping the qubits, and has no 1-local part, into the form required for Theorem \ref{thm:positive} using local unitaries (see~\cite{cubitt14} and references therein). Theorem \ref{thm:positive} thus classifies the complexity of all symmetric 2-qubit interactions with no 1-local part. We illustrate this classification in Figure \ref{fig:posweights}.

\begin{figure}
\begin{subfigure}[t]{0.3\textwidth}
\begin{center}
\begin{tikzpicture}[scale=0.5,>=stealth]
\fill[yellow!40] (-2,-2) rectangle (4,4);
\fill[red!25] (1,1) rectangle (4,4);
\draw[->] (-2,0) -- (5,0) node[anchor=west] {$\gamma$};
\draw[->] (0,-2) -- (0,5) node[anchor=south] {$\beta$};
\draw (-1,-0.2) -- (-1,0.2);
\draw (-0.2,-1) -- (0.2,-1);
\draw (1,-0.2) -- (1,0.2);
\draw (-0.2,1) -- (0.2,1);
\node at (-1,-0.5) {\footnotesize -1};
\node at (1,-0.5) {\footnotesize 1};
\node at (-0.5,1) {\footnotesize 1};
\node at (-0.5,-1) {\footnotesize -1};
\draw[ultra thick] (1,4) -- (1,1) -- (4,1);
\draw[ultra thick,green!50!black] (1,1) -- (-1,-1) -- (-2,-1);
\draw[ultra thick,green!50!black] (-1,-1) -- (-1,-2);
\node[green!50!black] at (-2.5,-1) {\footnotesize P};
\node[align=center] at (2.5,2.5) {\footnotesize QMA-\\ \footnotesize complete};
\node at (2.5,-1.5) {\footnotesize StoqMA};
\end{tikzpicture}
\end{center}
\caption{$H = -XX+\beta YY + \gamma ZZ$}
\end{subfigure}
\hfill
\begin{subfigure}[t]{0.3\textwidth}
\begin{center}
\begin{tikzpicture}[scale=0.5,>=stealth]
\fill[yellow!40] (-3,-3) rectangle (3,3);
\fill[red!25] (0,0) rectangle (3,3);
\draw[->] (-3,0) -- (4,0) node[anchor=west] {$\gamma$};
\draw[->] (0,-3) -- (0,4) node[anchor=south] {$\beta$};
\draw[ultra thick,blue] (0,3) -- (0,0) -- (3,0);
\draw[ultra thick,green!50!black] (0,-3) -- (0,0) -- (-3,0);
\node[align=center] at (1.5,1.5) {\footnotesize QMA-\\ \footnotesize complete};
\node at (-1.5,-1.5) {\footnotesize StoqMA};
\node[green!50!black] at (-3.5,0) {\footnotesize P};
\node[blue] at (-0.8,3.5) {\footnotesize NP};
\end{tikzpicture}
\end{center}
\caption{$H = \beta YY + \gamma ZZ$}
\end{subfigure}
\hfill
\begin{subfigure}[t]{0.3\textwidth}
\begin{center}
\begin{tikzpicture}[scale=0.5,>=stealth]
\fill[yellow!40] (-2,-2) rectangle (4,4);
\fill[red!25] (-1,4) -- (-1,1) -- (1,-1) -- (4,-1) -- (4,4) -- cycle;
\draw[->] (-2,0) -- (5,0) node[anchor=west] {$\gamma$};
\draw[->] (0,-2) -- (0,5) node[anchor=south] {$\beta$};
\draw[ultra thick] (-1,4) -- (-1,1) -- (1,-1) -- (4,-1);
\draw[ultra thick,green!50!black] (1,-1) -- (1,-2);
\draw[ultra thick,green!50!black] (-1,1) -- (-2,1);
\filldraw[blue] (0,0) circle (1mm);
\node[align=center] at (2,2) {\footnotesize QMA-\\ \footnotesize complete};
\node at (2.5,-1.5) {\footnotesize StoqMA};
\node[green!50!black] at (-2.5,1) {\footnotesize P};
\draw (-1,-0.2) -- (-1,0.2);
\draw (-0.2,-1) -- (0.2,-1);
\draw (1,-0.2) -- (1,0.2);
\draw (-0.2,1) -- (0.2,1);
\node at (-1,-0.5) {\footnotesize -1};
\node at (1,-0.5) {\footnotesize 1};
\node at (-0.5,1) {\footnotesize 1};
\node at (-0.5,-1) {\footnotesize -1};
\end{tikzpicture}
\end{center}
\caption{$H = XX+\beta YY + \gamma ZZ$}
\end{subfigure}
\caption{The complexity of interactions of the form $H = \alpha XX+\beta YY + \gamma ZZ$. The black line indicates StoqMA-complete interactions. The origin is in P in the first two cases, and NP-complete in the third.}
\label{fig:posweights}
\end{figure}
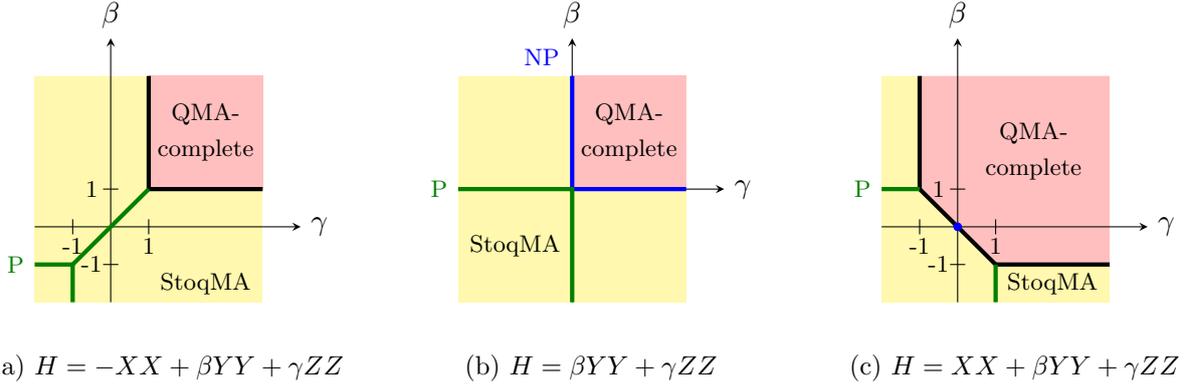

Important special cases of Theorem \ref{thm:positive} are the antiferromagnetic Heisenberg model ($H=XX+YY+ZZ$) and the antiferromagnetic XY model ($H=XX+YY$), which are both shown here to be QMA-complete for the first time. A similar result has been shown very recently for the antiferromagnetic XY interaction by Childs, Gosset and Webb~\cite{Childs2015}, using different techniques. Their result holds even when the $\alpha_i$ coefficients are all either 0 or 1, but only at fixed magnetisation (i.e.\ within a fixed eigenspace of $Z^{\otimes n}$). Another interesting case which does not fit into the framework of Theorem \ref{thm:positive} is the antiferromagnetic transverse Ising model, $\mathcal{S} = \{ZZ,X\}$. We show that this case of \spham\ is also StoqMA-complete.

The second restriction we study is to require that all interactions only occur on the edges of a particular interaction graph, where each vertex corresponds to a qubit. We focus on two natural graphs: 2D square and triangular lattices. Square lattices have been studied by several authors. It was shown by Oliveira and Terhal that general 2-local qubit Hamiltonians on a square lattice are QMA-complete~\cite{Oliveira2005}. Then the Heisenberg interaction with arbitrary 1-local interactions was shown to be QMA-complete on a square lattice by Schuch and Verstraete~\cite{Schuch2009}. It was shown in~\cite{cubitt14} that any QMA-complete 2-qubit interaction remains QMA-complete on the square lattice, if arbitrary 1-local terms are also allowed. However, no interaction has previously been shown to be QMA-complete on a square lattice without 1-local terms.

Our first result in this setting concerns square lattices without sign restrictions. To state the result, it will be helpful to decompose a general 2-qubit interaction $H$ into Pauli matrices: $H=\sum M_{ij} \sigma_i \otimes \sigma_j +$ 1-local terms. Then the Pauli rank of $H$ is the rank of the $3\times 3$ matrix $M$.

\begin{theorem}
\label{thm:square}
Let $H$ be a 2-local qubit interaction with Pauli rank at least 2, such that the 2-local part of $H$ is not proportional to $XX+YY+ZZ$. Then $\{H\}$\textsc{-Hamiltonian} is QMA-complete, even when the interactions are restricted to a 2D square lattice.
\end{theorem}

We believe that the one excluded case in this theorem, the Heisenberg interaction, is in fact QMA-complete; however, we have not been able to prove this. QMA-completeness of the Heisenberg interaction would follow if there existed an exactly solvable special case of the Heisenberg model satisfying certain constraints (see below). If this final case were found to be QMA-complete, we would have shown that all of the QMA-complete 2-qubit interactions remain QMA-complete on a square lattice.

Our second result about interaction geometry concerns 2D triangular lattices with sign restrictions. A square lattice can be viewed as a sublattice of a triangular lattice, and so any hardness results on a square lattice such as Theorem \ref{thm:square} can be trivially extended to the triangular lattice. However, we can obtain stronger results by working with triangular lattices directly.

\begin{theorem}
\label{thm:triangle}
Let $H=\alpha XX+\beta YY + \gamma ZZ$ be a 2-qubit interaction such that $\alpha+\beta>0,$ $\alpha+\gamma>0,\beta+\gamma>0$ and $H$ is not proportional to $XX+YY+ZZ$. Then $\{H\}^+$\textsc{-Hamiltonian} is QMA-complete, even if the interactions are restricted to the edges of a 2D triangular lattice.
\end{theorem}

An important special case of Theorem \ref{thm:triangle} is the antiferromagnetic XY model on a triangular lattice, which we show is QMA-complete. Observe that Theorems \ref{thm:positive} and \ref{thm:triangle} overlap, but Theorem \ref{thm:triangle} is not quite a stronger result than Theorem \ref{thm:positive}. We do not expect to be able to show QMA-completeness results for $\{H\}^+$\textsc{-Hamiltonian} for any of the interactions $H$ in Theorem \ref{thm:triangle} when restricted to a square lattice. Indeed, for any bipartite interaction graph, we can perform a local change of basis conjugating by $Z$ on every qubit on one side of the bipartition. This will effectively send $\alpha,\beta \rightarrow -\alpha,-\beta$, thereby converting any of the QMA-complete Hamiltonians from Theorem \ref{thm:positive} into one for which approximating the ground state energy is contained in StoqMA. Together, Theorems \ref{thm:square} and \ref{thm:triangle} thus highlight a way in which the triangular lattice can rigorously be said to be more complex than the square lattice.

\subsection{Proof techniques}

The main proof technique used is to apply perturbative gadgets~\cite{Kempe2004,Oliveira2005,Bravyi2014a,cubitt14} to show reductions from problems that are already known to be QMA-complete. The basic idea is to use a Hamiltonian $H=\Delta H_0 +V$, for some large $\Delta$, such that the low energy sector of $H$ is approximately equal to the ground space of $H_0$, and that within this space the action of $H$ is determined by $V$ to some order in a perturbative expansion in $1/\Delta$. This allows new interactions to be approximately implemented that were not previously available. For Theorem \ref{thm:positive}, a second order perturbation ``sign'' gadget is used to allow a Hamiltonian with positively weighted terms to simulate a Hamiltonian with mixed signs that is known to be QMA-complete. The StoqMA-completeness results are based on perturbative reductions from the transverse Ising model, which was recently shown to be complete for StoqMA by Bravyi and Hastings~\cite{Bravyi2014a}.

The main technical challenge in order to prove QMA-hardness for interactions on a 2D lattice turns out to be proving QMA-hardness of the XY model. This proof consists of two steps. It is known that Hamiltonians made up of XX and ZZ interactions on a square lattice, and arbitrary 1-local interactions, are QMA-complete. The first step of the proof develops a new variant of a reduction used in~\cite{cubitt14}, in order to reduce such Hamiltonians to an XY model Hamiltonian on a spatially sparse graph. A key technical ingredient in the gadget construction is a careful analysis of the exactly solvable antiferromagnetic XY model on a cyclic chain. The second step is to apply variants of gadgets used by Oliveira and Terhal~\cite{Oliveira2005} to reduce the sparse graph to a lattice. Finally, to prove QMA-hardness of other interactions, we define new gadgets for reductions from the XY model which fit onto lattices.

\subsection{Organisation}

We begin, in Section \ref{sec:perturb}, by describing the ideas from perturbation theory and the gadgets that we will use throughout the paper. Section \ref{sec:pws} contains our results on interactions with restricted signs and StoqMA-completeness. Section \ref{sec:restrict} describes the QMA-hardness proofs for 2D lattices, with calculations regarding the XY model on a cyclic chain deferred to Section \ref{sec:XYcyclic}. We conclude with some open problems in Section \ref{sec:outlook}.

\section{Perturbation theory}
\label{sec:perturb}

\subsection{Definition of simulation}

We typically start with a target Hamiltonian $H_{\text{target}}$ acting on an $N$-dimensional Hilbert space $\mathcal{H}$, of a form which is already known to be QMA-complete. The idea is to use a simulator Hamiltonian $H_{\text{sim}}$ which acts on a larger Hilbert space $\mathcal{H}_{\text{sim}}$, such that when restricted to its low energy subspace $\mathcal{L}_N(H_{\text{sim}})$ -- that is, the subspace spanned by the $N$ eigenvectors of $H_{\text{sim}}$ with lowest eigenvalues -- $H_{\text{sim}}$ is approximately the same as the target Hamiltonian $H_{\text{target}}$. We formalise this idea of simulation, using the definitions presented in~\cite{Bravyi2014a}.

We require that there exists an isometry $\widetilde{\mathcal{E}}:\mathcal{H} \rightarrow \mathcal{H}_{\text{sim}}$ such that the image of $\widetilde{\mathcal{E}}$ is exactly $\mathcal{L}_N(H_{\text{sim}})$ and $\widetilde{\mathcal{E}}^{\dagger}H_{\text{sim}}\widetilde{\mathcal{E}}$ approximates $H_{\text{target}}$ in operator norm up to a small error $\epsilon$. This is enough to show that the first $N$ eigenvalues of $H_{\text{sim}}$ approximate the eigenvalues of $H$ up to error $\epsilon$.

However, this gives us no information about the eigenvectors of $H_{\text{sim}}$ and how they relate to the eigenvectors of $H_{\text{target}}$, as the isometry $\widetilde{\mathcal{E}}$ is unknown and may in general be very complicated. It will therefore be useful to have a simple isometry $\mathcal{E}:\mathcal{H}\rightarrow\mathcal{H}_{\text{sim}}$ that approximates $\widetilde{\mathcal{E}}$ in operator norm. This motivates the following definition of simulation as put forward by Bravyi and Hastings~\cite{Bravyi2014a}:

\begin{definition}
Let $H$ be a Hamiltonian acting on a Hilbert space $\mathcal{H}$ of dimension $N$. A Hamiltonian $H_{\operatorname{sim}}$ and an isometry (encoding) $\mathcal{E}:\mathcal{H}\rightarrow\mathcal{H}_{\text{sim}}$ are said to simulate $H$ with error $(\eta,\epsilon)$ if there eists an isometry $\widetilde{\mathcal{E}}:\mathcal{H} \rightarrow \mathcal{H}_{\operatorname{sim}}$ such that
\begin{itemize}
\item[S1.] The image of $\widetilde{\mathcal{E}}$ coincides with the low-energy subspace $\mathcal{L}_N(H_{\text{sim}})$.
\item[S2.] $\|H-\widetilde{\mathcal{E}}^{\dagger}H_{\operatorname{sim}}\widetilde{\mathcal{E}}\| \leqslant \epsilon$.
\item[S3.] $\|\mathcal{E}-\widetilde{\mathcal{E}}\|\leqslant \eta$.
\end{itemize}
\end{definition}

It follows from Weyl's inequality~\cite{Bravyi2014a} that, if $(H_{\operatorname{sim}},\mathcal{E})$ simulates $H$ with error $(\eta,\epsilon)$, the $i$'th smallest eigenvalues of $H_{\operatorname{sim}}$ and $H$ differ by at most $\epsilon$ for all $1 \le i \le N$. In addition, it is shown in~\cite{Bravyi2014a} that simulation makes sense under composition as long as the lowest $N$ eigenvalues are separated from the rest of the spectrum by a relatively large gap.

\begin{lemma}[Bravyi and Hastings~\cite{Bravyi2014a}]
\label{lem:simulate}
Suppose $(H_1,\mathcal{E}_1)$ simulates $H$ with error $(\eta_1,\epsilon_1)$ and $(H_2,\mathcal{E}_2)$ simulates $H_1$ with error $(\eta_2,\epsilon_2)$. Let $\Delta_1$ be the spectral gap separating the $N$ smallest eigenvalues of $H_1$ from the rest of the spectrum and suppose that $\Delta_1 > 2\epsilon_2$ and $\epsilon_1,\epsilon_2 \leqslant \|H\|$. Then $(H,\mathcal{E}_2\mathcal{E}_1)$ simulates $H$ with error $(\eta,\epsilon)$, where $\eta = \eta_1 + \eta_2 + O(\epsilon_2 \Delta_1^{-1})$ and $\epsilon = \epsilon_1 + \epsilon_2 + O(\epsilon_2 \Delta_1^{-1} \|H\|)$.
\end{lemma}

If $\Delta_1 \gg \|H\|$, we have $\eta \approx \eta_1 + \eta_2$, $\epsilon \approx \epsilon_1 + \epsilon_2$. 

\subsection{Perturbative gadgets}

The idea of a perturbative gadget is to have a simulator Hamiltonian consisting of two parts $H_{\text{sim}}=\Delta H_0+V$, for some large parameter $\Delta \gg 1$, where $\|V\|\leqslant \Delta/2$, so that the $\Delta H_0$ term dominates. 

Let $\mathcal{H}_-$ be the $N$ dimensional ground space of $H_0$, and split the simulator space as $\mathcal{H}_{\text{sim}}=\mathcal{H}_- \oplus \mathcal{H}_+$. The encoding map $\mathcal{E}:\mathcal{H}\rightarrow\mathcal{H}_{\text{sim}}$ will always be chosen such that Im$(\mathcal{E})=\mathcal{H}_-$, and will usually be obvious in context. Let $P_{\pm}$ be the projection operators onto $\mathcal{H}_{\pm}$. For any linear operator $O$ acting on $\mathcal{H}_{\text{sim}}$ we define the following subscript notation:
\[O_{--}=P_-OP_- \quad O_+=P_+OP_+ \quad O_{-+}=P_-OP_+ \quad O_{+-}=P_+OP_-.\]
We assume that $(H_0)_{--}=0$ and that all eigenvalues of $(H_0)_{++}$ are greater than 1.

To find the effective Hamiltonian, we use the Schrieffer--Wolff transformation~\cite{bravyi11,Bravyi2014a}. This is a unitary operator $e^S$ on $\mathcal{H}_{\text{sim}}$, where $S$ is an antihermitian operator satisfying $S_{--}=0=S_{++}$. $e^S$ relates the subspaces $\mathcal{L}_N(H_{\text{sim}})$ and $\mathcal{H}_-$, such that $(e^S H_{\text{sim}}e^{-S})_{-+}=0=(e^S H_{\text{sim}}e^{-S})_{+-}$. These equations in fact determine the operator $S$ uniquely. Let $\widetilde{\mathcal{E}}=e^{-S}\mathcal{E}$ so (S1) is satisfied by the definition of $e^{-S}$, and for (S3) $\|\mathcal{E}-\widetilde{\mathcal{E}}\|=\|I-e^{-S}\|=O(\|S\|)$ which it is possible to show is small for large $\Delta$. 

Then define the effective Hamiltonian  $H_{\text{eff}}=(e^S H_{\text{sim}}e^{-S})_{--}$ , so
\[\widetilde{\mathcal{E}}^{\dagger}H_{\text{sim}}\widetilde{\mathcal{E}}=\mathcal{E}^{\dagger}e^SH_{\text{sim}}e^{-S}\mathcal{E}=\mathcal{E}^{\dagger}H_{\text{eff}}\mathcal{E}\]
and so (S2) becomes 
\[\|H_{\text{target}}-\widetilde{\mathcal{E}}^{\dagger}H_{\text{sim}}\widetilde{\mathcal{E}}\|=\|H_{\text{target}}-\mathcal{E}^{\dagger}H_{\text{eff}}\mathcal{E}\| = \|\overline{H}_{\text{target}}-H_{\text{eff}}\|\]
where $\overline{H}_{\text{target}}=\mathcal{E}H_{\text{target}}\mathcal{E}^{\dagger}$ is the logical encoding of the target Hamiltonian $H_{\text{target}}$ in the simulator space. 

The aim is therefore to find a simulator Hamiltonian $H_{\text{sim}}$ such that $H_{\text{eff}}$ approximates $\overline{H}_{\text{target}}$. Calculating $H_{\text{eff}}$ exactly is very difficult, but we can express it as a Taylor series in powers of $V$. Truncating this Taylor series at first, second and third order gives the following three lemmas respectively. For a more detailed explanation of the Schrieffer-Wolf transformation and proofs of all of these results, see~\cite{Bravyi2014a}.
\begin{lemma}[First Order]\label{lem:1st}~\cite{Bravyi2014a} Suppose one can choose $H_0, V$ such that the non-zero eigenvalues of $H_0$ are all greater than or equal to 1 and
\[\|\overline{H}_{\text{target}} - V_{--}\| \leqslant \varepsilon/2. \]
Suppose $\|V\| \leqslant \Lambda$. Then $H_{sim} = \Delta H_0 + V$ simulates
$H_{target}$ with error $(\eta,\epsilon)$, provided that $\Delta \geqslant \Omega(\epsilon^{-1}\Lambda^2+\eta^{-1}\Lambda).$
\end{lemma}

\begin{lemma}[Second Order] \label{lem:2nd}~\cite{Bravyi2014a} Suppose one can choose $H_0, V_{\text{main}}, V_{\text{extra}}$ such that $(H_0)_{++}$ has all eigenvalues greater than or equal to 1, $(V_{\text{extra}})_{+-} = (V_{\text{extra}})_{-+} = 0$, $(V_{\text{main}})_{--} = 0$, and
\[\| \overline{H}_{\text{target}} - (V_{\text{extra}})_{--} + (V_{\text{main}})_{-+}H_{0}^{-1}(V_{\text{main}})_{+-}\| \leqslant \varepsilon/2. \]
Suppose $\|V_{\text{main}}\|, \|V_{\text{extra}}\| \leqslant \Lambda$. Then $H_{sim} = \Delta H_0 + \Delta^{1/2} V_{\text{main}} + V_{\text{extra}}$ simulates
$H_{target}$ with error $(\eta,\epsilon)$, provided that $\Delta \geqslant \Omega(\epsilon^{-2}\Lambda^6+\eta^{-2}\Lambda^2).$
\end{lemma}

\begin{lemma}[Third Order] \label{lem:3rd}~\cite{Bravyi2014a} Suppose one can choose $H_0, V_{\text{main}}, V_{\text{extra}},\widetilde{V}_{\text{extra}} $ such that $(H_0)_{++}$ has all eigenvalues greater than or equal to 1, $(V_{\text{extra}})_{+-} = (V_{\text{extra}})_{-+} = 0$, $(\widetilde{V}_{\text{extra}})_{+-} = (\widetilde{V}_{\text{extra}})_{-+} = 0$, $(V_{\text{main}})_{--} = 0,$\
\[\|\overline{H}_{\text{target}} - (V_{\text{extra}} )_{--} + (V_{\text{main}})_{-+}H_{0}^{-1}(V_{\text{main}})_{++}H_{0}^{-1}(V_{\text{main}})_{+-}\| \leqslant \epsilon/2\]
and
\[(\widetilde{V}_{\text{extra}})_{--}=(V_{\text{main}})_{-+}H_0^{-1}(V_{\text{main}})_{+-}.\]
Suppose $\|V_{\text{main}}\|, \|V_{\text{extra}}\|,\|\widetilde{V}_{\text{extra}}\|\leqslant \Lambda$. Then $H_{sim} = \Delta H_0 + \Delta^{2/3} V_{\text{main}} + \Delta^{1/3}\widetilde{V}_{\text{extra}}+V_{\text{extra}}$ simulates
$H_{target}$ with error $(\eta,\epsilon)$, provided that $\Delta \geqslant \Omega(\epsilon^{-3}\Lambda^{12}+\eta^{-3}\Lambda^3)$.
\end{lemma}

Lemmas \ref{lem:simulate}--\ref{lem:3rd} reduce the problem of simulating desired effective interactions to the problem of finding suitable operators $H_0, V_{\text{main}}, V_{\text{extra}},\widetilde{V}_{\text{extra}}$. To prove hardness results we will always take $\epsilon,\eta=O(1/\poly(n))$, making the approximation error small enough that it can be disregarded. We do not attempt to optimise (or even calculate precisely) the size of the weights $\Delta=O(\poly(n))$ that occur in our reductions, but they will always be polynomial in $n$ and easily computable.

\subsection{Parallel use of gadgets}
\label{sec:parallel}
It is often required to implement a poly($n$) number of perturbative gadgets in parallel across different subsets of qubits without increasing the coefficients to super-polynomial size. Here we describe (following ideas of~\cite{Oliveira2005,bravyi11}) how this can be achieved in a quite general setting.

Consider a system of $n$ qubits partitioned into $k+1$ disjoint subsets of qubits, labelled $S_i$ for $i=1,2,\dots k$. We will apply heavy interactions to the first $k$ subsets, in order to simulate an effective Hamiltonian on the remaining subset of qubits $S_{k+1}$.  Let $H_0^{(i)}$ denote a heavy interaction which acts only on $S_i$, $V^{(i)}$ be an interaction term that acts non trivially only on $S_i \cup S_{k+1}$, and $P_{-}^{(i)}$ the projection operator onto the ground space of $H_0$. 

Now consider the Hamiltonian terms for the whole system $H_0=\sum_{i}H_0^{(i)}$ and $V=\sum_{i} V^{(i)}$.
The projection operator onto the ground space of $H_0$ is given by $\prod_{i} P_{-}^{(i)}$. 
For a first order perturbation it is immediate that $V_{--}=\sum_{i} V^{(i)}_{--}$, and so applying Lemma \ref{lem:1st}, we see that $\Delta H_0 +V$ simulates $\sum V_{--}^{(i)}$, as long as $\Delta = \Omega (\epsilon^{-1} \|V\|^2 )$.

Since $H_{\text{eff}}$ acts only on $\mathcal{H}_-$, we need only consider how such an operator acts on the ground space of the overall Hamiltonian $H_0$, which has basis elements that can be expressed as a product state $|\psi^{(1)}\>|\psi^{(2)}\>\dots |\psi^{(k)}\>|\phi\>$ where each $|\psi^{(i)}\>$ is in the ground space of $H_0^{(i)}$.

For second order perturbation terms $V^{(i)}$, we need to calculate $V_{-+}H_0^{-1}V_{+-}$. Since $V^{(i)}$ acts nontrivially only on subsets $i$ and $k+1$, only $|\psi^{(i)}\>$ will be elevated out of the ground space by $V^{(i)}_{+-}$, and so:
  \[V_{-+}H_0^{-1}V_{+-}=\sum_{i,j}V^{(j)}_{-+}H_0^{-1}V^{(i)}_{+-}=\sum_{i}V^{(i)}_{-+}(H^{(i)}_0)^{-1}V^{(i)}_{+-}\]

For third order perturbation terms $V^{(i)}$ satisfying $V^{(i)}_{--}=0$, each $V^{(i)}$ term can only move the $i$th subset of qubits in and out of the ground space $\mathcal{H}_-$, so
\[V_{-+}H_0^{-1}V_{++}H_0^{-1}V_{+-}=\sum_{i,j,k}V^{(i)}_{-+}H_0^{-1}V^{(j)}_{++}H_0^{-1}V^{(k)}_{+-}=\sum_{i}V^{(i)}_{-+}(H^{(i)}_0)^{-1}V^{(i)}_{++}(H^{(i)}_0)^{-1}V^{(i)}_{+-}\]
So in each case, the effective Hamiltonian simulated by $V=\sum V^{(i)}$ and $H_0=\sum H_0^{(i)}$  is the sum of the interactions simulated by each of the $V^{(i)}$ and $H_0^{(i)}$ separately.

\subsection{The basic gadget}
\label{sec:ex}

We now describe a second-order perturbative gadget which will be used throughout this paper. Let $H=\alpha XX+\beta YY+\gamma ZZ$ where $\alpha+\beta>0,$ $\alpha+\gamma>0$ and $\beta+\gamma>0$. Consider a system of $n$ qubits labelled by $i \in \{1,2,\dots n\}$ along with two extra ancilla qubits $a$ and $b$. The heavy interaction term $H_0$ is just the interaction $H$ applied to the ancilla qubits and a constant term to ensure that $(H_0)_{--}=0$. This could then be multiplied by a further constant term to ensure that $(H_0)_{++}$ has all eigenvalues greater than 1, as required for Lemmas \ref{lem:1st}-\ref{lem:3rd}. Then
\begin{equation} \label{eq:h0} H_0=\tfrac{1}{2}\left[ H_{ab}+(\alpha+\beta+\gamma)I\right]
=(\alpha+\beta)|\Psi^{+}\>\<\Psi^{+}| +(\alpha+\gamma)|\Phi^{+}\>\<\Phi^{+}|+(\beta+\gamma)|\Phi^{-}\>\<\Phi^{-}|\end{equation}
where $|\Psi^{\pm}\>, |\Phi^{\pm}\>$ are the maximally entangled states defined by
 $|\Psi^{\pm}\>=\frac{|01\>\pm|10\>}{\sqrt{2}}, \quad |\Phi^{\pm}\>=\frac{|00\>\pm|11\>}{\sqrt{2}}$.

In this diagonal form, we see that the conditions $\alpha+\beta>0,$ $\alpha+\gamma>0$ and $\beta+\gamma>0$ are exactly what is required for $|\Psi^-\>$ to be the unique ground state of $H$. Furthermore it is easy to calculate the inverse of $H$, which will be necessary to determine the second order perturbation:

\[H_0^{-1}=\frac{1}{\alpha+\beta}|\Psi^{+}\>\<\Psi^{+}| +\frac{1}{\alpha+\gamma}|\Phi^{+}\>\<\Phi^{+}|+\frac{1}{\beta+\gamma}|\Phi^{-}\>\<\Phi^{-}|. \]

To define $V$, partition the qubits $\{1,2,\dots n\}$ into two sets $A$ and $B$; and let 
\[V=\sum_{i \in A} \lambda_i H_{ia} +\sum_{j \in B} \mu_j H_{jb}. \]

Note that since $|\Psi^-\>_{ab}$ is maximally entangled, $V_{--}=0$. We also need to calculate $V_{-+}$, which can be done term by term:
\[(H_{ia})_{-+}=\alpha X_i |\Psi^- \>\<\Phi^{-}|_{ab}+i\beta Y_i |\Psi^-\>\<\Phi^+|_{ab} -\gamma Z_i |\Psi^-\>\<\Psi^+|_{ab}=-(H_{ib})_{-+}\]
To calculate $V_{+-}$, simply note that for any Hermitian operator $V$, $V_{+-}=(V_{-+})^{\dagger}$. This is now enough to calculate the second order perturbations. First
\[(H_{ia})_{-+} H_0^{-1} (H_{ia})_{+-}=\left(\frac{\alpha ^2}{\beta+\gamma}+\frac{\beta ^2}{\alpha+\gamma}+\frac{\gamma ^2}{\alpha+\beta}\right)|\Psi^-\>\<\Psi^-|_{ab}=(H_{ib})_{-+} H_0^{-1} (H_{ib})_{+-}\]
These terms just contribute to an overall energy shift in the Hamiltonian. The following terms simulate more interesting interactions:
\[(H_{ia})_{-+} H_0^{-1} (H_{jb})_{+-}=\left(-\frac{\alpha ^2}{\beta+\gamma}X_iX_j-\frac{\beta ^2}{\alpha+\gamma}Y_iY_j-\frac{\gamma ^2}{\alpha+\beta}Z_iZ_j\right)|\Psi^-\>\<\Psi^-|_{ab}\]
\[(H_{ia})_{-+} H_0^{-1} (H_{ja})_{+-}=\left(\frac{\alpha ^2}{\beta+\gamma}X_iX_j+\frac{\beta ^2}{\alpha+\gamma}Y_iY_j+\frac{\gamma ^2}{\alpha+\beta}Z_iZ_j\right)|\Psi^-\>\<\Psi^-|_{ab}\]
\[(H_{ib})_{-+} H_0^{-1} (H_{jb})_{+-}=\left(\frac{\alpha ^2}{\beta+\gamma}X_iX_j+\frac{\beta ^2}{\alpha+\gamma}Y_iY_j+\frac{\gamma ^2}{\alpha+\beta}Z_iZ_j\right)|\Psi^-\>\<\Psi^-|_{ab}\]

We define
\begin{equation} \label{eq:htilde} \widetilde{H}=\frac{\alpha ^2}{\beta+\gamma}XX+\frac{\beta ^2}{\alpha+\gamma}YY+\frac{\gamma ^2}{\alpha+\beta}ZZ \end{equation}
and will use this notation throughout the paper. Lemma \ref{lem:2nd} tells us that for large $\Delta$, the Hamiltonian $\Delta H_0 +\Delta^{1/2}V$ simulates, up to an overall energy shift, the following target Hamiltonian:

\[H_{\text{target}}=\sum_{
i \in A, j \in B } \lambda_i\mu_j \widetilde{H}_{ij}  -
\sum_{i,j \in A} \lambda_i\lambda_j \widetilde{H}_{ij} - \sum_{i,j \in B} \mu_i \mu_j \widetilde{H}_{ij} \] 

The gadget is illustrated in Figure \ref{fig:ex}. We see that a positive interaction is simulated between any two target qubits connected to opposite ancilla qubits and a negative interaction is simulated between any two target qubits connected to the same ancilla qubit.

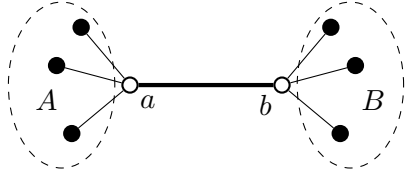
\begin{figure}[t]
\centering
\begin{tikzpicture} 

\draw[dashed] (-0.9,0) ellipse (0.7cm and 1.1cm);
\draw[dashed] (2.9,0) ellipse (0.7cm and 1.1cm);
\node at (-1.1,-0.2) {$A$};
\node at (3.2,-0.2) {$B$};

\node[mqubit] (a) at (0,0) {};
\node[mqubit] (b) at (2,0){};
\draw[heavy]  (a) to (b);

\node[qubit] (i1) at ($(a)+(130:1)$){};
\node[qubit] (i2) at ($(a)+(165:1)$){};
\node[qubit] (i3) at ($(a)+(-140:1)$){};
\draw  (i1) to (a);
\draw  (i2) to (a);
\draw  (i3) to (a);

\node[qubit] (j1) at ($(b)+(50:1)$){};
\node[qubit] (j2) at ($(b)+(15:1)$){};
\node[qubit] (j3) at ($(b)+(-40:1)$){};
\draw  (j1) to (b);
\draw  (j2) to (b);
\draw  (j3) to (b);

\node[anchor=north west] at (0,0) {$a$};
\node[anchor=north east] at (2,0) {$b$};
\end{tikzpicture}

\caption{Interaction graph for the basic gadget. In this figure, as throughout the paper, vertices represent qubits and edges represent interactions. Vertices which are filled in represent qubits between which we would like to produce effective interactions, while vertices not filled in represent ancilla qubits. Thick lines represent heavy interactions, while thin lines represent weak interactions.}
\label{fig:ex}
\end{figure}

\subsubsection{Antisymmetric case}
\label{sec:antisym}

We will also use a similar gadget for antisymmetric interactions $H$. A general 2-local antisymmetric interaction has the form $XZ-ZX$ up to normalisation. We will use the following basis, in which $XZ-ZX$ is diagonal:
\[|\psi_0\>=\frac{|00\>+|01\>-|10\>+|11\>}{2}, \quad |\psi_1\>=\frac{|00\>+|01\>+|10\>-|11\>}{2}\]
\[|\psi_2\>=\frac{-|00\>+|01\>+|10\>+|11\>}{2}, \quad |\psi_3\>=\frac{|00\>-|01\>+|10\>+|11\>}{2}.\]

We consider a Hamiltonian with the same interaction graph as in Figure \ref{fig:ex}. As before, the heavy interaction term $H_0$ acts on a pair of ancilla qubits labelled $a$ and $b$, where
\[H_0=\tfrac{1}{2}(X_aZ_b-Z_aX_b+2I)=|\psi_1\>\<\psi_1|+|\psi_2\>\<\psi_2|+2|\psi_3\>\<\psi_3|.\]
And as before the other terms will be of the form
\[V=\sum_{i \in A} \lambda_i (X_aZ_i-Z_aX_i) +\sum_{j \in B} \mu_j (X_bZ_j-Z_bX_j).\]

 Note that $(X_aZ_i-Z_aX_i)_{--}=0=(X_bZ_i-Z_bX_i)_{--}$ for all $i$, so $V_{--}=0$. To work out the second order terms, we first calculate
\[(X_aZ_i-Z_aX_i)_{+-}=Z_i|\psi_2\>\<\psi_0|_{ab}-X_i|\psi_1\>\<\psi_0|_{ab}\]
\[(X_bZ_i-Z_bX_i)_{+-}=Z_i|\psi_1\>\<\psi_0|_{ab}+X_i|\psi_2\>\<\psi_0|_{ab}\]
So the second order terms are
\[(X_aZ_i-Z_aX_i)_{-+}H_0^{-1}(X_aZ_i-Z_aX_i)_{+-}=2|\psi_0\>\<\psi_0|_{ab}= (X_bZ_i-Z_bX_i)_{-+}H_0^{-1}(X_bZ_j-Z_bX_j)_{+-}\]
which just contribute an overall energy shift, and 
\[(X_aZ_i-Z_aX_i)_{-+}H_0^{-1}(X_aZ_j-Z_aX_j)_{+-}=(X_iX_j+Z_iZ_j)|\psi_0\>\<\psi_0|_{ab}\]
\[(X_aZ_i-Z_aX_i)_{-+}H_0^{-1}(X_bZ_j-Z_bX_j)_{+-}=(Z_iX_j-X_iZ_j)|\psi_0\>\<\psi_0|_{ab}\]
which allows us to simulate an $XX+ZZ$ interaction between qubits connected to the same ancilla qubit, and a $ZX-XZ$ interaction between qubits connected to opposite ancilla qubits. That is, by Lemma 2, the Hamiltonian $H_{\text{sim}}=\Delta H_0 +\Delta^{1/2}V$ simulates the target Hamiltonian
\[H_{\text{target}}=\sum_{
i \in A, j \in B } \lambda_i\mu_j (X_iZ_j-Z_iX_j)  -
\sum_{i,j \in A} \lambda_i\lambda_j (X_iX_j+Z_iZ_j) - \sum_{i,j \in B} \mu_i \mu_j (X_iX_j+Z_iZ_j). \] 

\section{Positive weights}
\label{sec:pws}

We now have all the ingredients we need to prove Theorem \ref{thm:positive}.

\begin{lemma}
\label{lem:positive1}
Given a purely 2-local symmetric interaction $H$, which can be written in the form $\alpha XX+\beta YY + \gamma ZZ$, the problem $\{H\}^+$-\textsc{Hamiltonian} is either:
\begin{description}
\item[i)] QMA-complete, if $\alpha+\beta>0,$ $\alpha+\gamma>0$ and $\beta+\gamma>0$
\item[ii)] in StoqMA, otherwise.
\end{description}
\end{lemma}

\begin{proof}[\textbf{\emph{Proof}}]
First we show that for case ii) the Hamiltonian $H$ is stoquastic, in the correct choice of basis, and hence that the problem $\{H\}^+$-\textsc{Hamiltonian} is in StoqMA. Indeed, by a local change of basis on each qubit, we can relabel $X,Y,$ and $Z$ interchangeably. Therefore without loss of generality we can assume that $\alpha\leqslant\beta\leqslant\gamma$. We have
\[H=\alpha XX+ \beta YY + \gamma ZZ= \left( \begin{array}{cccc}
\gamma & 0 & 0 & \alpha-\beta \\
0 & -\gamma & \alpha+\beta & 0 \\
0 & \alpha+\beta & -\gamma & 0 \\
\alpha-\beta & 0 & 0 & \gamma \\
\end{array}\right). \]
Looking at the matrix elements of $H$, the only off-diagonal elements that could be positive are those equal to $\alpha+\beta$, but since we are not in case i) and $\alpha\leqslant\beta\leqslant\gamma$, we must have that $\alpha+\beta \leqslant 0$. So all off-diagonal elements are non-positive, which is the definition of a stoquastic Hamiltonian.

Now we consider case i). We show the problem is QMA-complete, by showing we can use the basic gadget from Section \ref{sec:ex} to simulate any Hamiltonian which is a sum of terms of the form 
\[\widetilde{H}=\frac{\alpha ^2}{\beta+\gamma}XX+\frac{\beta ^2}{\alpha+\gamma}YY+\frac{\gamma ^2}{\alpha+\beta}ZZ\]
with both positive and negative weights. Observe that, as $\alpha+\beta>0,$ $\alpha+\gamma>0$ and $\beta+\gamma>0$,  $\widetilde{H}$ is a sum of at least 2 Pauli terms. The \lham\ problem was shown in~\cite{cubitt14} to be QMA-complete for any interaction of this form.

For each desired interaction between two of the $n$ target qubits we wish to simulate, we add a pair of ancillary qubits, resulting in a total of at most $n+2n(n-1)/2=n^2$ physical qubits. Suppose for example that we wished to simulate an interaction between target qubits 1 and 2. Then we introduce a pair of ancilla qubits $a_{12}$ and $b_{12}$ and apply a heavy $\Delta H_{a_{12}b_{12}}$ interaction between them.
Then if we wish to simulate a positive interaction we use a second order perturbation term $V=H_{1a_{12}}+H_{b_{12}2}$; or for a negative interaction we use $V=H_{1a_{12}}+H_{a_{12}2}$. As shown in Section \ref{sec:ex}, this will simulate the interactions $\widetilde{H}_{12}$ and $-\widetilde{H}_{12}$ respectively.

These gadgets can then all be applied in parallel, as discussed in Section  \ref{sec:parallel}.
\end{proof}

Of those interactions contained in StoqMA, some of them are trivially contained in P. In particular, this holds when $|00\>$ is the ground state of $H$. In this case the state $|000\dots 0\>$ is the ground state for any Hamiltonian of the form $H_{\text{tot}}=\sum_{ij}\lambda_{ij}H_{ij}$,  with energy $\<00|H|00\>\sum_{ij} \lambda_{ij}$.

Since $|\Phi^{\pm}\>$ and $|\Psi^{\pm}\>$ are always eigenstates with energies as in eqn.\ (\ref{eq:h0}), the state $|00\>$ is a ground state when both $|\Phi^+\>$ and $|\Phi^-\>$ have the smallest energy. That is, when
\[\alpha+\gamma=\beta+\gamma \leqslant \min(0,\alpha+\beta)\]
which is equivalent to
\[\alpha=\beta \quad \text{and} \quad \gamma \leqslant -|\alpha|.\]
Therefore for any interaction $H$ which is, up to normalisation and relabelling $X,Y,Z$, of the form $\alpha (XX+YY)-ZZ$ for $|\alpha|\leqslant 1$, the problem $\{H\}^+$-\textsc{Hamiltonian} is contained in P.

\subsection{Antiferromagnetic TIM}

It remains to prove the StoqMA-completeness part of Theorem \ref{thm:positive}. Before we do so, we prove a related but somewhat simpler result about the Transverse Ising Model (TIM). This model consists of local interaction terms of the form $ZZ$ and $X$:
 \[H=\sum_{i,j} \alpha_{ij} Z_i Z_j +\sum_k \beta_k X_k\] 
It may be assumed that $\beta_k$ is non-negative for all $k$, by conjugating all qubits for which $\beta_k<0$ by $Z$. This leaves the $ZZ$ part of the Hamiltonian unchanged, and flips the sign of the affected $\beta_k$. The problem of finding the ground-state energy of this Hamiltonian was shown to be StoqMA-complete when both positive and negative signs are allowed for $\alpha_{ij}$, by Bravyi and Hastings~\cite{Bravyi2014}. However Bravyi~\cite{Bravyi2014a} also showed that if $\alpha_{ij}\leqslant 0$ for all $i,j$, then there is a polynomial-time probabilistic classical algorithm to find the ground-state energy (and even to approximate the partition function), so the problem $\{-ZZ,X\}^+$-\textsc{Hamiltonian} is in BPP.

We therefore focus on the antiferromagnetic case, where $\alpha_{ij}\geqslant 0$ for all $i,j$.

\begin{theorem}
The problem $\{ZZ,X\}^+$-\textsc{Hamiltonian} is StoqMA-complete.
\end{theorem}

\begin{proof}[\textbf{\emph{Proof}}] 
Since the full TIM Hamiltonian has been shown to be StoqMA-complete it will suffice to show that we can simulate this Hamiltonian using only positive $\alpha_{ij}$. We do this by encoding logical qubits in the two dimensional ground space of $ZZ$, which is span$\{|01\>,|10\>\}$.
We label logical states and operators with a $L$ superscript, identifying $|0^L\>=|\Psi^-\>$ and $|1^L\>=|\Psi^+\>$, and associate each logical qubit $i$ with physical qubits $i_a$ and $i_b$. We will use operators of the following form to simulate $\pm X^L X^L$ and $Z^L$.

\[H_0=\sum_{i=1}^{n}\tfrac{1}{2}\left(Z_{i_a}Z_{i_b}+ I\right) =\sum_{i=1}^{n}|00\>\<00|_i +|11\>\<11|_i\]

\[V_{\text{extra}}=\sum_{i,j:\lambda_{ij}>0}\lambda_{ij} Z_{i_a}Z_{j_a}+\sum_{i,j:\lambda_{ij}<0}-\lambda_{ij} Z_{i_a}Z_{j_b}\]
\[V_{\text{main}}=\sum_{i}\mu_{i} \left(X_{i_a}+X_{i_b}\right)\]

Calculating the first order perturbations gives:
\[(Z_{i_a}Z_{j_a})_{--}= (|\Psi^+\>\<\Psi^-|_i+|\Psi^-\>\<\Psi^+|_i)\otimes (|\Psi^+\>\<\Psi^-|_j+|\Psi^-\>\<\Psi^+|_j)=X^L_i X^L_j\]
\[(Z_{i_a}Z_{j_b})_{--}= (|\Psi^+\>\<\Psi^-|_i+|\Psi^-\>\<\Psi^+|_i)\otimes (-|\Psi^+\>\<\Psi^-|_j-|\Psi^-\>\<\Psi^+|_j)=- X^L_i X^L_j.\]
This is how it is possible to simulate both $XX$ and $-XX$. Note that $( X_{i_a})_{--}=0=( X_{i_b})_{--}$ for all $i$, so $V_{\text{main}}$ has no first order contribution. To calculate the second order perturbations for $V_{\text{main}}$, we first show
\[\left( X_{i_a}+X_{i_b}\right)_{+-}=(|00\>+|11\>)(\<01|+\<10|)_i=2|\Phi^+\>\<\Psi^+|_i\]
so 
\[\left( X_{i_a}+X_{i_b}\right)_{-+}H_0^{-1} \left( X_{i_a}+X_{i_b}\right)_{+-}= 4|\Psi^+\>\<\Psi^+|_i=2(I-Z^L_i)\]
Therefore, by Lemma \ref{lem:2nd}, the Hamiltonian $H_{\text{sim}}=\Delta H_0+\Delta^{1/2}V_{\text{main}}+V_{\text{extra}}$ simulates the target Hamiltonian:
\[H_{\text{target}}=\sum_{i,j} \lambda_{ij} X_i X_j + \sum_{i} 2\mu_{i}(I-Z_i). \]
Relabelling $X$ and $Z$ gives the general TIM Hamiltonian, which was shown to be StoqMA-complete in~\cite{Bravyi2014}.
\end{proof}

\subsection{The case $H=XX-YY+\gamma ZZ$ for $\gamma >1$}
We finally show that the remaining cases on the boundary of the StoqMA region are also StoqMA-complete, also using a reduction from TIM. The proof is similar to the previous section.

In the case where $\alpha=1=-\beta$ and $\gamma>1$, then both $|\Psi^-\>$ and $|\Psi^+\>$ are ground states for the interaction $H$. Therefore we can take a system of $2n$ qubits labelled by $i_a$ and $i_b$ for $i=1,2,\dots n$, and let 
\[H_0=\sum_{i=1}^{n}\tfrac{1}{2}\left(H_{i_a,i_b}+\gamma I\right) =\sum_{i=1}^{n}(\gamma+1)|\Phi^+\>\<\Phi^+|_i +(\gamma-1)|\Phi^-\>\<\Phi^-|_i\]

\[V_{\text{extra}}=\frac{1}{\gamma}\sum_{i,j:\lambda_{ij}>0}\lambda_{ij} H_{i_a,j_a}+\frac{1}{\gamma}\sum_{i,j:\lambda_{ij}<0}-\lambda_{ij} H_{i_a,j_b}\]
\[V_{\text{main}}=\sum_{i,j}\mu_{ij} \left(H_{i_a,j_a}+H_{i_a,j_b}\right)\]
The encoding $\mathcal{E}$ encodes a logical qubit into each pair of qubits, with $|0^L\>=|\Psi^-\>$ and $|1^L\>=|\Psi^+\>$.

Calculating the perturbations to first order gives, similarly to before,
\[(H_{i_a,j_a})_{--}=\gamma (|\Psi^+\>\<\Psi^-|_i+|\Psi^-\>\<\Psi^+|_i)\otimes (|\Psi^+\>\<\Psi^-|_j+|\Psi^-\>\<\Psi^+|_j)=\gamma X^L_i X^L_j\]
\[(H_{i_a,j_b})_{--}=\gamma (|\Psi^+\>\<\Psi^-|_i+|\Psi^-\>\<\Psi^+|_i)\otimes (-|\Psi^+\>\<\Psi^-|_j-|\Psi^-\>\<\Psi^+|_j)=-\gamma X^L_i X^L_j.\]
Therefore $(V_{\text{extra}})_{--}=\sum_{i,j} \lambda_{ij} X^L_i X^L_j $, and $(V_{\text{main}})_{--}=0$.
To calculate the second order contribution of $V_{\text{main}}$, consider just one term, $V=H_{i_a,j_a}+H_{i_a,j_b}$:
\[V_{+-}=2(|\Phi^+\>|\Phi^-\>-|\Phi^-\>|\Phi^+\>)\<\Psi^-|\<\Psi^+| +2(|\Phi^+\>|\Phi^+\>-|\Phi^-\>|\Phi^-\>)\<\Psi^+|\<\Psi^+| \]

Since $|\Phi^+\>|\Phi^-\>$ and $|\Phi^-\>|\Phi^+\>$ both have energy $2\gamma$ with respect to the Hamiltonian $H_0$; and $|\Phi^+\>|\Phi^+\>$ and $|\Phi^-\>|\Phi^-\>$ have energies $2\gamma+2$ and $2\gamma-2$ respectively, we get:

\[V_{-+}H_0^{-1}V_{+-}= \left(\frac{1}{2\gamma}+\frac{1}{2\gamma}\right)4|\Psi^-\>|\Psi^+\>\<\Psi^-|\<\Psi^+|+\left(\frac{1}{2\gamma+2}+\frac{1}{2\gamma-2}\right)4|\Psi^+\>|\Psi^+\>\<\Psi^+|\<\Psi^+|\]
\[=\frac{4}{\gamma}|0^L 1^L\>\<0^L 1^L|+\frac{4\gamma}{\gamma^2-1}|1^L 1^L\>\<1^L 1^L|=\frac{1}{\gamma(\gamma^2-1)}\left(Z^L-(2\gamma^2-1)I\right)\otimes (Z^L-I)\]

Then by Lemma \ref{lem:2nd}, the total Hamiltonian $H_{\text{sim}}=H_0+\Delta V_{\text{extra}}+\Delta^{1/2} V_{\text{main}}$ simulates the Hamiltonian 
\[H_{\text{target}}=\sum_{i,j} \lambda_{ij} X_i X_j - \mu_{ij}^2 \frac{1}{\gamma(\gamma^2-1)}\left(Z-(2\gamma^2-1)I\right)_i (Z-I)_j\]
and so we can implement any Hamiltonian which consists of interactions $XX$ and $F(\gamma)=-(Z-(2\gamma^2-1)I)(Z-I)$. Simulating a TIM Hamiltonian requires one further round of perturbation theory to simulate $Z$ interactions. We add an extra pair of ancilla qubits with a heavy $F(\gamma)$ interaction applied between them to project them into the ground state $|11\>$. Then we apply $F(\gamma)$ between one of these ancilla qubits and another qubit, to simulate a $Z$ interaction to this qubit.

This completes the proof of Theorem \ref{thm:positive}.

\section{Restricted interaction graphs}
\label{sec:restrict}

We now move on to proving QMA-completeness results for interactions with restricted geometries. First, we show that the XY model is QMA-complete even when the interactions are restricted to the edges of a triangular lattice, and have positive weights. To prove this we first use similar methods to~\cite{cubitt14} to show that the model is QMA-complete on a spatially sparse graph, as defined by Oliveira and Terhal in~\cite{Oliveira2005}. Then we generalise the fork and crossing gadgets used in~\cite{Oliveira2005} to reduce this to a 2D triangular lattice.

\subsection{Spatial sparsity}
In~\cite{cubitt14}, it is observed that the results of~\cite{Oliveira2005} and~\cite{biamonte08}  can be straightforwardly combined to show that a Hamiltonian of $XX$, $ZZ$, $X$ and $Z$ terms is QMA-complete even when the interaction terms are restricted to the edges of a 2D square lattice. We will simulate such a Hamiltonian using a variant of the gadget used in~\cite{cubitt14} to show that the XY model is QMA-complete. Unlike~\cite{cubitt14}, our gadget will fit onto a spatially sparse graph.

\begin{definition}[Oliveira and Terhal~\cite{Oliveira2005}]
A spatially sparse interaction graph $G$ is defined as a graph in
which (i) every vertex participates in $O(1)$ edges, (ii) there is a straight-line drawing in the
plane such that every edge overlaps with $O(1)$ other edges and the length of every edge is $O(1)$.
\end{definition}

Let $H = XX + YY$. In the construction of~\cite{cubitt14}, three physical qubits are used to encode two logical qubits in the ground space of $H_{12}+H_{23}+H_{13}$, which is the asymmetric subspace of three qubits.
 It is then shown that $XI$ and $ZI$ can be simulated by applying further interactions within this subspace, and that applying interactions across triples of physical qubits can result in logical interactions of the form:
 \[X_iX_j (X_{i'}X_{j'}+Y_{i'}Y_{j'})\]
\[Z_iZ_j (X_{i'}X_{j'}+Y_{i'}Y_{j'})\]
 \[I_iI_j (X_{i'}X_{j'}+Y_{i'}Y_{j'})\]
where we have labelled the first and second qubit in each pair of logical qubits as $i$ and $i'$ respectively. We can use the third type of interaction term to project all of the second qubits in each pair into the ground state of a certain Hamiltonian $H_{set}$. Having done this we can use the first two terms to create a Hamiltonian of the form:
\[\sum_{k} (\alpha_k X_k +\beta_k Z_k) + \sum_{i<j} (\gamma_{ij} X_iX_j +\delta_{ij}Z_iZ_j ) \<\Omega|X_{i'}X_{j'}+Y_{i'}Y_{j'}|\Omega\>\]
where $|\Omega\>$ is the ground state of $H_{set}$.

Therefore it is possible to implement any Hamiltonian consisting of terms $\{XX,ZZ,X,Z\}$, as long as we can find a Hamiltonian $H_{set}$ such that:
\begin{enumerate}
\item The ground state $|\Omega\>$ is unique;
\item The spin correlation functions $ \<\Omega|X_{i'}X_{j'}+Y_{i'}Y_{j'}|\Omega\>$ can be calculated in poly($n$) time on a classical computer, so that the $\gamma_{ij}, \delta_{ij}$ coefficients can be suitably adjusted;
\item In order that $\gamma_{ij}, \delta_{ij}$ are $O(\text{poly}(n))$, it is also necessary that the correlation functions are $\Omega(1/\text{poly}(n))$;
\item Similarly we need the spectral gap of $H_{set}$ to be $\Omega(1/\text{poly}(n))$.
\end{enumerate}

In the original proof the XY model on the complete graph was used~\cite{cubitt14}, but this does not satisfy the spatial sparsity constraints. To meet this constraint, we instead consider a 1D model. In this context simultaneously satisfying conditions (3) and (4) is challenging. We require $\Omega(1/\text{poly}(n))$ size correlation functions for qubits at large distance apart on the chain. But, as it is known that for any system of this form with a constant spectral gap the correlation functions will decay exponentially~\cite{Hastings2006}, we require that the spectral gap $\rightarrow 0$. However, by (4) we also require that the gap does not go to zero too quickly.

Fortunately, however, the XY model on a cyclic chain satisfies all of these constraints:

\begin{lemma}
\label{lem:xy}
Fix $N$ even but not a multiple of 4, and let $H = \sum_{i=1}^{N-1} (X_i X_{i+1} + Y_i Y_{i+1}) + X_1 X_N +Y_1 Y_N$. Then $H$ has a nondegenerate ground state $|\Omega\>$ and spectral gap $\Omega(1/N)$. Further, for any pair $i$, $j$ such that $|i-j| = n$ and $n = o(N^{4/7})$, $\<\Omega|X_{i}X_{j}+Y_{i}Y_{j}|\Omega\> = \Omega(n^{-1/2})$. There is an efficient classical algorithm to compute the spectral gap and all the correlation functions.
\end{lemma}

The proof of Lemma \ref{lem:xy} is deferred to Section \ref{sec:XYcyclic}. Most of the claims in the lemma are well-known~\cite{Lieb1961,McCoy68}; however, we were unable to find in the literature a rigorous lower bound on the correlation functions for a finite-length chain which is as tight as we need.

To simulate a Hamiltonian $H$ which consists of $n$ qubits on a square lattice, we first add $\sqrt{n}$ ancilla qubits, which are not involved in any interactions, along one edge of the lattice (see Figure \ref{fig:XYsquare}). Then, for each of the qubits in this extended lattice, we use three physical qubits to encode two logical qubits and apply the Hamiltonian of the XY model on a cyclic chain to the second qubit in each pair with the path of the cycle arranged as in Figure \ref{fig:XYsquare}. Because the ancilla qubits did not interact with any others in the original lattice, the longest distance along this path for which we need to calculate the correlation functions (in order to simulate $XX$ and $ZZ$ interactions) is $2n^{1/2}$. By Lemma \ref{lem:xy}, the correlation functions are $\Omega(n^{-1/2})$.

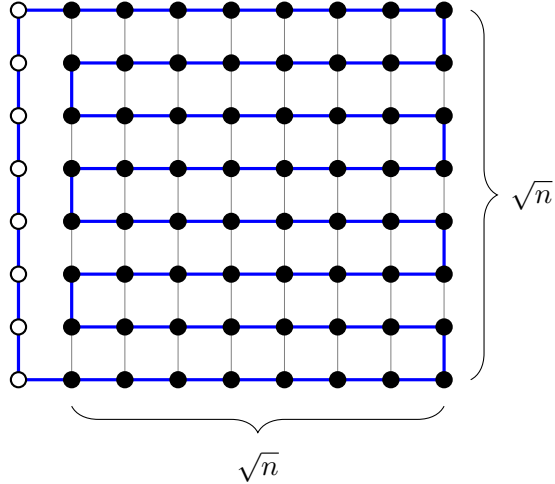
\begin{figure}[t]
\centering
\begin{tikzpicture}[scale=0.7]
\def\n{8}
\draw[step=1cm,gray,very thin] (1,1) grid (\n,\n);
\foreach \y in {1,...,\n}
{\foreach \x in {2,...,\n}
{\draw[blue,very thick] (\x,\y) to (\x-1,\y);}}

\draw[blue,very thick] (0,1) to (1,1);
\draw[blue,very thick] (0,\n) to (1,\n);
\foreach \y in {2,...,\n}
{\draw[blue,very thick] (0,\y) to (0,\y-1);}

\foreach \y in {1,3,...,\n}
{\draw[blue,very thick] (\n,\y) to (\n,\y+1);}
\foreach \y in {3,5,...,\n}
{\draw[blue,very thick] (1,\y) to (1,\y-1);}

\foreach \y in {1,...,\n}
{
\node[mqubit] (0\y) at (0,\y){};
\node[qubit] (1\y) at (1,\y){};
\foreach \x in {2,...,\n}
{
\node[qubit] (\x\y) at (\x,\y){};
}}

\draw [decorate,decoration={brace,amplitude=10pt,mirror}]
(\n+0.5,1) -- (\n+0.5,\n) node [black,midway,xshift=+0.8cm] {$\sqrt{n}$};

\draw [decorate,decoration={brace,amplitude=10pt,mirror}]
(1,0.5) -- (\n,0.5) node [black,midway,yshift=-0.8cm] {$\sqrt{n}$};
\end{tikzpicture}
\caption{Cyclic XY model Hamiltonian overlaid on a 2D square lattice, with $\sqrt{n}$ ancilla qubits (coloured white).}
\label{fig:XYsquare}
\end{figure}

We have now obtained a spatially sparse Hamiltonian which only uses XY interactions and simulates $H$. However, the interactions no longer take place on a square lattice, because each site on the lattice contains 3 physical qubits in a triangle, and there are complicated interactions between neighbouring triangles. To go back to a lattice geometry we need some further gadgetry.

\subsection{Mediator qubit pair gadgets}
\label{sec:medgadg}
In \cite{Oliveira2005}, Oliveira and Terhal develop ``subdivision'', ``fork'' and ``crossing'' gadgets involving a \emph{mediator} qubit. This is an ancilla qubit $a$, which is projected into the state $|0\>$ by a heavy $I-Z$ interaction, such that when 2-local terms are applied, such as $V=H_{1a}+H_{2a}$, second order perturbation theory gives an effective interaction between qubits 1 and 2, even though there is no direct interaction between them. This allows one to simulate a general Hamiltonian on a spatially sparse interaction graph using a Hamiltonian on a planar 2D graph.

However, here we only have access to 2-local terms. The direct analogue of these gadgets is to take a pair of qubits, $a$ and $b$ and project these into the ground state of a 2-local interaction $H_0=H_{ab}$. The second order perturbation theory analysis of this situation is exactly the basic gadget described in Section \ref{sec:ex}; and so the interactions generated will be of the form $\widetilde{H}$ (see eqn.\ (\ref{eq:htilde})). We focus on the XY interaction, where $H=XX+YY$ and so $H=\widetilde{H}$ (this is similar for the Heisenberg interaction $H=XX+YY+ZZ$, so all of these gadgets would work in this case too). For the fork and crossing gadgets, where there are three or more extra qubits, unwanted interaction terms are generated that need to be cancelled out by first order terms.

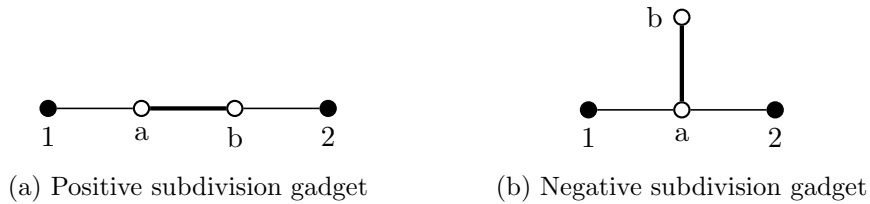
\begin{figure}[b]
\centering
\begin{subfigure}[b]{0.4\textwidth}
\centering
\begin{tikzpicture} 
\node[qubit] (q1) [label=below:1] {};
\node[mqubit] (a) [right=of q1,label=below:a] {};
\node[mqubit] (b) [right=of a,label=below:b]{};
\node[qubit] (q2) [right=of b,label=below:2]{};

\draw [heavy] (a) to (b);
\draw [black, semithick] (q1) to (a);
\draw [black, semithick] (q2) to (b);
\end{tikzpicture}
\caption{Positive subdivision gadget}
\end{subfigure}
\begin{subfigure}[b]{0.4\textwidth}
\centering
\begin{tikzpicture} 
\node[qubit] (q1) [label=below:1] {};
\node[mqubit] (a) [right=of q1,label=below:a] {};
\node[mqubit] (b) [above=of a,label=left:b]{};
\node[qubit] (q2) [right=of a,label=below:2]{};

\draw [heavy] (a) to (b);
\draw [black, semithick] (q1) to (a);
\draw [black, semithick] (q2) to (a);
\end{tikzpicture}
\caption{Negative subdivision gadget}
\end{subfigure}
\caption{Subdivision gadgets}
\label{fig:sub}
\end{figure}

\subsubsection{Subdivision gadgets}
\label{sec:subdivision}
The simplest gadgets are the subdivision gadgets, which are in a sense the most basic gadgets we can make using the XY interaction. The interaction graphs for the gadgets are shown in Figure \ref{fig:sub}. There are two additional qubits, labelled 1 and 2, as well as the mediator qubits, which are either connected to the same or opposite ancilla qubits. 

When connected to opposite ancilla qubits, i.e. $V=H_{1a}+H_{2b}$, the analysis of Section \ref{sec:ex} shows that $\Delta H_0 + \Delta^{1/2}V$ simulates a $X_1X_2+Y_1Y_2$ interaction. 
And when connected to the same ancilla qubit, i.e. $V=H_{1a}+H_{2a}$, the analysis of Section \ref{sec:ex} shows that $\Delta H_0 + \Delta^{1/2}V$ simulates a $-X_1X_2-Y_1Y_2$ interaction. 

The positive subdivision gadget can be applied in series, that is it can be used to simulate any of the interactions $H_{1a}$, $H_{ab}$, and $H_{b2}$, using a new heavy weight $\Delta' \gg \Delta$ to give a Hamiltonian on 6 qubits in a line, with an effective $XX+YY$ interaction between the first and last qubits. This can be repeated at most a constant number of times giving a Hamiltonian on $2k$ qubits on a line, where $k=O(1)$, which simulates an $XX+YY$ interaction between the qubits at the endpoints.

Similarly for the negative gadget, both $H_{1a}$ and $H_{2a}$ can be simulated by the positive subdivision gadget, and this process can be repeated $O(1)$ times. Thus for any $k=O(1)$, there exists a Hamiltonian on $2k$ qubits where $2k-1$ qubits are on a line, and one extra qubit is connected to a qubit at an odd internal position on this line, which simulates a $-XX-YY$ interaction between the qubits at the endpoints of the line.

\subsubsection{Fork gadget}
The fork gadget is used to reduce the degree of a vertex in an interaction graph. For a qubit interacting with two other qubits, we can simulate this using a gadget which only has one incoming edge to this qubit. The interaction graph is given in Figure \ref{fig:fork}, and has $V_{\text{main}}=H_{1a}+H_{2a}+H_{3b}$. This will also generate an unwanted $-H_{12}$, so it is necessary to also include the term $V_{\text{extra}}=H_{12}$, which is represented on the interaction graph by a dotted line. Then, by Lemma \ref{lem:2nd}, the Hamiltonian $H_{\text{sim}} = \Delta H_0 + \Delta^{1/2} V_{\text{main}} + V_{\text{extra}}$ simulates $H_{13}+H_{23}$.

\subsubsection{Crossing gadget}
The crossing gadget is used to simulate a Hamiltonian in which two edges in the interaction graph cross, using a gadget with no crossings, as can be seen in Figure \ref{fig:crossing}. The main second order terms are $V_{\text{main}}=H_{1a}+H_{2a}+H_{3b}+H_{4b}$, and the required correction terms are $V_{\text{extra}}=H_{12}+H_{34}-H_{13}-H_{24}$. Then, by Lemma \ref{lem:2nd}, the Hamiltonian $H_{\text{sim}} = \Delta H_0 + \Delta^{1/2} V_{\text{main}} + V_{\text{extra}}$ simulates $H_{14}+H_{23}$.

\begin{figure}
\centering

\begin{subfigure}[b]{0.4\textwidth}
\begin{tikzpicture} 
\node[qubit] (q1) at (-1,3) [label=below:1] {};
\node[qubit] (q2) at (1,3)[label=below:2]{};
\node[mqubit] (a) at (0,2)[label=left:a] {};
\node[mqubit] (b) at (0,1)[label=left:b]{};
\node[qubit] (q3) at (0,0)[label=below:3]{};

\draw [heavy] (a) to (b);
\draw [black, semithick] (q1) to (a);
\draw [black, semithick] (q2) to (a);
\draw [black, semithick] (q3) to (b);
\draw [dotted] (q1) to (q2);

\node[qubit] (1) at (-5,2) [label=below:1] {};
\node[qubit] (2) at (-3,2)[label=below:2]{};
\node[qubit] (3) at (-4,1)[label=below:3]{};

\draw [black, semithick] (1) to (3);
\draw [black, semithick] (2) to (3);
\draw [dotted] (2) to (3);
\end{tikzpicture}
\caption{Fork gadget}
\label{fig:fork}
\end{subfigure}
\hspace{1cm}
\begin{subfigure}[b]{0.4\textwidth}
\centering
\begin{tikzpicture} 
\node[qubit] (q1) at (-1,3) [label=left:1] {};
\node[qubit] (q2) at (1,3)[label=right:2]{};
\node[mqubit] (a) at (0,2)[label=left:a] {};
\node[mqubit] (b) at (0,1)[label=left:b]{};
\node[qubit] (q3) at (-1,0)[label=below:3]{};
\node[qubit] (q4) at (1,0)[label=below:4]{};

\draw [heavy] (a) to (b);
\draw [black, semithick] (q1) to (a);
\draw [black, semithick] (q2) to (a);
\draw [black, semithick] (q3) to (b);
\draw [black, semithick] (q4) to (b);
\draw [dotted] (q1) to (q2);
\draw [dotted] (q2) to (q4);
\draw [dotted] (q3) to (q4);
\draw [dotted] (q3) to (q1);

\node[qubit] (1) at (-5,3) [label=left:1] {};
\node[qubit] (2) at (-3,3)[label=right:2]{};
\node[qubit] (3) at (-5,0)[label=below:3]{};
\node[qubit] (4) at (-3,0)[label=below:4]{};
\draw [black, semithick] (1) to (4);
\draw [black, semithick] (2) to (3);
\end{tikzpicture}
\caption{Crossing gadget}
\label{fig:crossing}
\end{subfigure}

\caption{Fork and crossing gadgets. In each case the left-hand interaction pattern is simulated using the right-hand gadget.}
\label{fig:forkcross}
\end{figure}
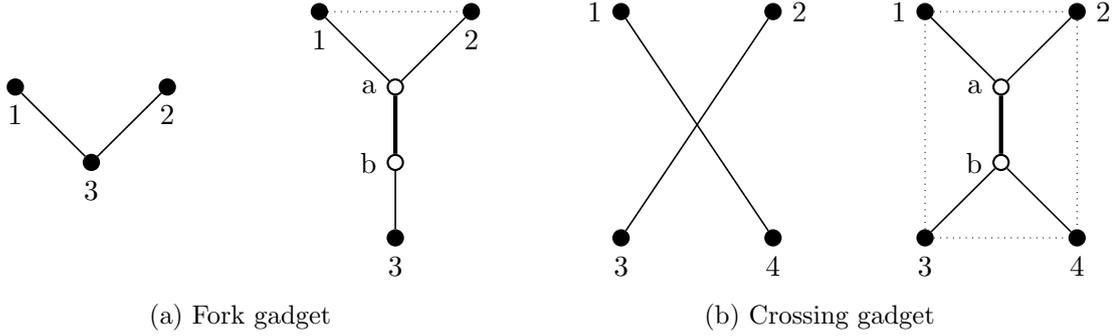

\subsection{XY model on a 2D planar graph of degree 3}
\label{sec:planar}

We can now follow the prescription of Oliveira and Terhal~\cite{Oliveira2005}, using the analogous mediator qubit gadgets described in Section \ref{sec:medgadg}, to simulate a spatially sparse Hamiltonian using a triangular lattice. This follows~\cite{Oliveira2005} closely so we only sketch the steps involved.

First use the subdivision gadget on every single edge in parallel in order to isolate each vertex of high degree. Then, for each vertex of degree $d=O(1)$, pair up adjacent incoming edges and use the fork gadget in parallel on each pair of incoming edges, to reduce the degree of the vertex to $\lceil d/2 \rceil$. Then repeat this process $O(\log d)$ times in series until the vertex has degree at most 3. Each stage of this process can be done in parallel across all vertices of high degree, such that only a constant number of reductive steps are required to end up with an interaction graph of degree 3.

Then we isolate each crossing by applying the subdivision gadget up to $O(\log m)$ times on each edge in series, where $m=O(1)$ is the maximum number of crossings across any edge in the graph. Finally we apply the crossing gadget to every crossing in parallel to be left with a Hamiltonian with a planar interaction graph $G$ of degree at most 3, such that there is a straight-line drawing in the plane where each edge is of length $O(1)$ and the angle between adjacent incoming edges at a vertex is $\Omega(1)$.

\subsection{2D triangular lattice}
We finally simulate the Hamiltonian of the previous section using a Hamiltonian on a 2D triangular lattice with only positive weights.

To do this we take the interaction graph $G$ from the previous section, along with its embedding in the 2D plane such that the graph is planar with maximum degree 3, each edge has length $O(1)$, and the angle between adjacent edges is $\Omega(1)$. Then we place a fine triangular grid on top of this graph and move each vertex of $G$ onto the nearest lattice point, then deform edges of $G$ onto the nearest path through the lattice. Since the angle between any two edges is $\Omega(1)$,  if the spacing in the lattice is sufficiently small (but still $\Omega(1)$), these paths will only intersect in a small $O(1)$ region around each original vertex. The paths can then easily be rerouted in this area.

As discussed in Section \ref{sec:medgadg}, the subdivision gadget can be used in parallel to simulate $+H$ (resp.\ $-H$) when there are an even (resp.\ odd) number of mediator qubits on the path in between the target qubits. On a triangular lattice it is easy to include one extra qubit on the path, for example by taking two edges of a triangle rather than one, and so it is possible to choose the appropriate parity, and simulate either $+H$ or $-H$ as desired. This completes the proof that the antiferromagnetic XY interaction on a triangular lattice is QMA-complete.

Note that this construction will work for essentially any non-bipartite lattice. A graph is non-bipartite if and only if it contains an odd cycle. If the spacing of the lattice is small enough such that there is such an odd cycle between each connected pair of original vertices, and this odd cycle is connected to the rest of the graph at two points, then two paths that pass through these two points (but which go around this cycle in opposite directions) will have lengths of opposite parities, and thus we can make either $+H$ or $-H$ as desired.

\subsection{Other interactions on a triangular lattice}

We now return to considering more general interactions of the form $H = \alpha XX + \beta YY + \gamma ZZ$. In order to prove that these interactions are QMA-complete on a triangular lattice, we will simulate a target Hamiltonian of XY interactions on a triangular lattice using only local gadgets. We require that when each edge of the target interaction graph is replaced by the local gadget, the resulting graph (the interaction graph for the simulator Hamiltonian) is still a triangular lattice. Therefore each local gadget must fit on a triangular lattice, and in order that separate gadgets do not interfere, they must be able to tessellate without overlapping. Figure \ref{fig:tessel} shows that the triangular gadget described below does indeed tessellate in this way.

The idea behind the gadget is to simulate two different interactions $H^{(1)}$ and $H^{(2)}$ along two different paths between two qubits, and then to take a linear combination of these to access a different interaction known to be QMA-complete, such as the XY interaction. Figure \ref{fig:triangle}(b) shows such a gadget that will simulate two different interactions between qubits 1 and 6 (where $\widetilde{H}$ is as in (\ref{eq:htilde}) and $H'$ is defined below). However, both qubits 1 and 6 have two incoming edges, and so this gadget will not tessellate on the lattice. We therefore need to use an equivalent of the fork gadget so that qubits 1 and 6 only have one incoming edge each.

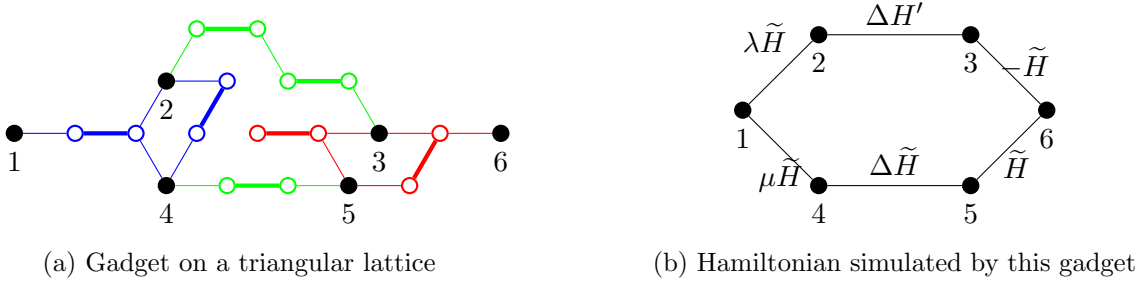
\begin{figure}
\centering
\begin{subfigure}[b]{0.4\textwidth}
\centering
\begin{tikzpicture}[scale=0.8]
\begin{scope}[blue]
\node[qubit,black] (1) at (0,0)[label={[black]below:1}]{};
\node[mqubit] (a1) at ($(1)+(1,0)$){};
\node[mqubit] (a2) at ($(a1)+(1,0)$){};
\node[qubit,black] (2) at ($(a2)+(60:1)$)[label={[black]below:2}]{};
\node[qubit,black] (4) at ($(a2)+(-60:1)$)[label={[black]below:4}]{};
\node[mqubit] (b1) at ($(a2)+(1,0)$){};
\node[mqubit] (b2) at ($(2)+(1,0)$){};

\draw (1) to (a1);
\draw (2) to (a2);
\draw (2) to (b2);
\draw (4) to (a2);
\draw (4) to (b1);
\draw[heavy] (a1) to (a2);
\draw[heavy] (b1) to (b2);

\end{scope}

\begin{scope}[red]
\node[mqubit] (c1) at ($(b1)+(1,0)$){};
\node[mqubit] (c2) at ($(c1)+(1,0)$){};
\node[qubit,black] (3) at ($(c2)+(1,0)$)[label={[black]below:3}]{};
\node[qubit,black]  (5) at ($(c2)+(-60:1)$)[label={[black]below:5}]{};
\node[mqubit] (d1) at ($(5)+(1,0)$){};
\node[mqubit] (d2) at ($(d1)+(60:1)$){};
\node[qubit,black] (6) at ($(d2)+(1,0)$)[label={[black]below:6}]{};

\draw (6) to (d2);
\draw (3) to (d2);
\draw (3) to (c2);
\draw (5) to (d1);
\draw (5) to (c2);
\draw[heavy] (c1) to (c2);
\draw[heavy] (d1) to (d2);
\end{scope}

\begin{scope}[green]
\node[mqubit] (e1) at ($(2)+(60:1)$){};
\node[mqubit] (e2) at ($(e1)+(1,0)$){};
\node[mqubit] (e3) at ($(e2)+(-60:1)$){};
\node[mqubit] (e4) at ($(e3)+(1,0)$){};
\draw (2) to (e1);
\draw (e3) to (e2);
\draw (e4) to (3);
\draw[heavy] (e1) to (e2);
\draw[heavy] (e3) to (e4);

\node[mqubit] (f1) at ($(4)+(1,0)$){};
\node[mqubit] (f2) at ($(f1)+(1,0)$){};
\draw (4) to (f1);
\draw (f2) to (5);
\draw[heavy] (f1) to (f2);
\end{scope}
\end{tikzpicture}
\caption{Gadget on a triangular lattice}
\end{subfigure}
\hspace{2cm}
\begin{subfigure}[b]{0.4\textwidth}
\centering
\begin{tikzpicture}
\node[qubit] (1) at (0,0)[label={[black]below:1}]{};
\node[qubit] (2) at (1,1)[label={[black]below:2}]{};
\node[qubit] (3) at (3,1)[label={[black]below:3}]{};
\node[qubit] (4) at (1,-1)[label={[black]below:4}]{};
\node[qubit] (5) at (3,-1)[label={[black]below:5}]{};
\node[qubit] (6) at (4,0)[label={[black]below:6}]{};
\draw (1) -- (2) node [near start,above=10pt] {$\lambda\widetilde{H}$};
\draw (2) -- (3) node [midway,above] {$\Delta H'$};
\draw (3) -- (6) node [near end,above] {$-\widetilde{H}$};
\draw (1) -- (4) node [midway,below] {$\mu\widetilde{H}$};
\draw (4) -- (5) node [midway,above] {$\Delta \widetilde{H}$};
\draw (5) -- (6) node [near start,right] {$\widetilde{H}$};
\end{tikzpicture}
\caption{Hamiltonian simulated by this gadget}
\end{subfigure}
\caption{Triangular lattice gadget for simulating $XX+YY$}
\label{fig:triangle}
\end{figure}

The interaction graph of the triangular lattice gadget is shown in Figure \ref{fig:triangle}(a), and it will simulate the Hamiltonian shown in Figure \ref{fig:triangle}(b). The red and blue sections of the graph are generalizations of the Fork gadget for interactions of the form $H=\alpha XX+\beta YY+\gamma ZZ$. Each of them consists of two parallel applications of the basic gadget. Consider the fork gadget coloured blue in Figure \ref{fig:triangle}; the mediator qubit pair connected to qubit 1 generates interactions of the form $+\widetilde{H}_{12}$ and $+\widetilde{H}_{14}$, but also an unwanted  interaction term $-\widetilde{H}_{24}$ between qubits 2 and 4. This is cancelled out by simulating a $+\widetilde{H}_{24}$ interaction with the other pair of mediator qubits between qubits 2 and 4. If both of the physical $H$ interactions acting on qubit 2 in this fork gadget have weight $\lambda$, and both interactions acting on qubit 4 have weight $\mu$, then the overall effect of the blue fork gadget is to simulate $\lambda\widetilde{H}_{12} +\mu\widetilde{H}_{14}$.

Similarly the part of the gadget coloured red can be viewed as a kind of fork gadget. One mediator qubit pair simulates $-\widetilde{H}_{36},+\widetilde{H}_{56}$  and $+\widetilde{H}_{35}$ interactions, while the other pair of mediator qubits simulates $-\widetilde{H}_{35}$ resulting in an overall effective Hamiltonian of $\widetilde{H}_{56}-\widetilde{H}_{36}$.

The gadget between qubits 4 and 5 is just the basic gadget and so simulates $\widetilde{H}_{45}$, but the gadget between qubits 2 and 3 is slightly different and needs to be studied separately. Labelling the qubits in between 2 and 3 as $a,b,c,d$, the heavy Hamiltonian for this gadget is $H_0=H_{ab}+H_{cd}+2(\alpha+\beta+\gamma)I$ which has non-degenerate ground state $|\Psi^{-}\>_{ab}|\Psi^-\>_{cd}$. The other terms are $V=H_{2a}+H_{bc}+H_{d3}$, and it is straightforward to show that
\[V_{--}=0, \quad \text{and} \quad V_{-+}H_{0}^{-1}V_{+-} \propto I\]
To calculate the third order contribution, it will help to write out $V$ as nine terms,
\[ V=\alpha X_2X_a+\beta Y_2Y_a+\gamma Z_2Z_a+\alpha X_bX_c+\beta Y_bY_c+\gamma Z_bZ_c+\alpha X_dX_3+\beta Y_dY_3+\gamma Z_dZ_3, \]
and work in the basis $|\Psi^{\pm}\>,|\Phi^{\pm}\>$ for the qubit pairs $(a,b)$ and $(c,d)$. Then each term of $V$ maps basis elements to basis elements, and $H_0$ is diagonal.
\[ V_{-+}H_{0}^{-1}V_{++}H_{0}^{-1}V_{+-}=\frac{\alpha ^3}{(\beta+\gamma)^2}X_2X_3+\frac{\beta ^3}{(\alpha+\gamma)^2}Y_2Y_3+\frac{\gamma ^3}{(\alpha+\beta)^2}Z_2Z_3\]

Therefore, by Lemma \ref{lem:3rd}, the Hamiltonian $\Delta H_0+\Delta^{2/3}V$ simulates the interaction $H'$:

\[H'=\frac{\alpha ^3}{(\beta+\gamma)^2}XX+\frac{\beta ^3}{(\alpha+\gamma)^2}YY+\frac{\gamma ^3}{(\alpha+\beta)^2}ZZ=:\alpha'XX+\beta'YY+\gamma'ZZ\]

Note that the coefficients of $H'$ satisfy $\alpha'+\beta'>0$, $\beta'+\gamma'>0$ and $\alpha'+\gamma'>0$, so $H'$ can be used as the heavy interaction term $H_0$ in a gadget similar to that of Section \ref{sec:ex}.

Therefore, overall the gadget simulates the Hamiltonian shown in Figure \ref{fig:triangle}(b). This Hamiltonian is then used to simulate an interaction $-\lambda H^{(1)}$ (via qubits 2 and 3) and an interaction $+\mu H^{(2)}$ (via qubits 4 and 5), where

\[H^{(1)}=\frac{\widetilde{\alpha}^2}{\beta'+\gamma'} XX+ \frac{\widetilde{\beta}^2}{\alpha'+\gamma'}YY +\frac{\widetilde{\gamma}^2}{\alpha'+\beta'}ZZ\]
\[H^{(2)}=\frac{\widetilde{\alpha}^2}{\widetilde{\beta}+\widetilde{\gamma}} XX+ \frac{\widetilde{\beta}^2}{\widetilde{\alpha}+\widetilde{\gamma}}YY +\frac{\widetilde{\gamma}^2}{\widetilde{\alpha}+\widetilde{\beta}}ZZ.\]
We now show that having access to interactions of the form $-H^{(1)}$ and $H^{(2)}$, with arbitrary positive weights, is sufficient to produce an $XX+YY$ interaction.

\subsubsection{The case $\alpha XX+ \beta YY$}
\label{sec:abtri}
First we consider the case where $H$ has Pauli rank 2, where $\gamma=0$, and so $\alpha,\beta>0$. In this case it is particularly easy to explicitly find $\widetilde{\alpha},\alpha',\widetilde{\beta},\beta'$ in terms of $\alpha$ and $\beta$ to show that
\[ H^{(1)}=\frac{\alpha^6}{\beta^5}XX+\frac{\beta^6}{\alpha^5}YY\quad \text{ and } \quad H^{(2)}=\frac{\alpha^5}{\beta^4}XX+\frac{\beta^5}{\alpha^4}YY. \]
In this case we can directly simulate  the $XX+YY$ interaction with $-\lambda H^{(1)}+\mu H^{(2)}$, by taking 
\[\lambda=\frac{\alpha^9-\beta^9}{\alpha^4\beta^4(\alpha^2-\beta^2)} \quad \text{ and } \quad \mu  =\frac{\alpha^{11}-\beta^{11}}{\alpha^5\beta^5(\alpha^2-\beta^2)}. \]
The exact values that $\lambda$ and $\mu$ take are not too significant; but it is important to note that they are both positive (given that $\alpha,\beta >0$) and are easily computable.

\subsubsection{The case $\alpha XX+\beta YY +\gamma ZZ$}
For more general interactions $H$, where all of $\alpha,\beta,\gamma$ are non-zero, the simulated interactions $H^{(1)}$ and $H^{(2)}$ are guaranteed to be different unless $H$ is the Heisenberg interaction $XX+YY+ZZ$.
In order to normalise $-\lambda H^{(1)}+\mu H^{(2)}$, let $\lambda=1-\mu$, and define $\alpha(\mu),\beta(\mu),\gamma(\mu)$ to be the coefficients of 
\[H(\mu)=-H^{(1)}+\mu(H^{(1)}+H^{(2)}). \]
For $\mu \in [0,1]$, this is the one parameter family of interactions that can be simulated using the gadget. So $H(0)=-H^{(1)}$ which has all coefficients negative and $H(1)=H^{(2)}$ which has all coefficients positive. Therefore there exists $\mu_{\alpha}, \mu_{\beta}, \mu_{\gamma} \in (0,1)$ such that $\alpha(\mu_{\alpha})=0, \beta(\mu_{\beta})=0, \gamma(\mu_{\gamma})=0$. In particular:
\[\mu_{\alpha}=\frac{\alpha^{(1)}}{\alpha^{(1)}+\alpha^{(2)}}, \quad \mu_{\beta}=\frac{\beta^{(1)}}{\beta^{(1)}+\beta^{(2)}}, \quad
\mu_{\gamma}=\frac{\gamma^{(1)}}{\gamma^{(1)}+\gamma^{(2)}}\]

Calculating the coefficients $\alpha(\mu),\beta(\mu),\gamma(\mu)$ exactly in terms of $\alpha,\beta,\gamma$ and $\mu$ yields very messy expressions, but the following lemma gives a very useful relation between $\alpha,\beta,\gamma$ and $\mu_{\alpha},\mu_{\beta},\mu_{\gamma}$.

\begin{lemma}
\label{lem:mu}
If $\gamma=\alpha$, then $\mu_{\gamma}=\mu_{\alpha}$. If $\gamma>\alpha$ and $\gamma \geqslant \beta$, then $\mu_{\gamma}>\mu_{\alpha}$.
\end{lemma}

\begin{proof}
The aim is to show that $\mu_{\gamma}-\mu_{\alpha}\geqslant 0$. First simply substitute in the expressions for $\mu_{\gamma}$ and $\mu_{\alpha}$:
\[\mu_{\gamma}-\mu_{\alpha} =
\frac{\gamma^{(1)}}{\gamma^{(1)}+\gamma^{(2)}} -\frac{\alpha^{(1)}}{\alpha^{(1)}+\alpha^{(2)}} =\frac{\alpha^{(2)}\gamma^{(1)}-\alpha^{(1)}\gamma^{(2)}}{(\gamma^{(1)}+\gamma^{(2)})(\alpha^{(1)}+\alpha^{(2)})}\]
Since every term in the denominator is positive, it will suffice to consider just the numerator.
\[\alpha^{(2)}\gamma^{(1)}-\alpha^{(1)}\gamma^{(2)}=\frac{\widetilde{\alpha}^2\widetilde{\gamma}^2}{(\beta'+\gamma')(\widetilde{\beta}+\widetilde{\gamma})(\alpha'+\beta')(\widetilde{\alpha}+\widetilde{\beta})} \left[(\beta'+\gamma')(\widetilde{\alpha}+\widetilde{\beta})-(\widetilde{\beta}+\widetilde{\gamma})(\alpha'+\beta')\right]\]
Again, the factor outside the square brackets is strictly positive, so we can just consider the expression inside the square brackets. The following relation will be useful:
\[\beta'\widetilde{\alpha}-\widetilde{\beta}\alpha'=\frac{\alpha^2\beta^2}{(\alpha+\gamma)(\beta+\gamma)}\left(\frac{\beta}{\alpha+\gamma}-\frac{\alpha}{\beta+\gamma}\right)=\frac{\alpha^2\beta^2(\alpha+\beta+\gamma)(\beta-\alpha)}{(\alpha+\gamma)^2(\beta+\gamma)^2}\]
The expression in the square brackets $(\beta'\widetilde{\alpha}-\widetilde{\beta}\alpha')+(\gamma'\widetilde{\alpha}-\widetilde{\gamma}\alpha')+(\gamma'\widetilde{\alpha}-\widetilde{\gamma}\alpha')$ is then just
\[\frac{\alpha+\beta+\gamma}{(\alpha+\gamma)^2(\beta+\gamma)^2(\alpha+\beta)^2}\left[\alpha^2\beta^2(\alpha+\beta)^2(\beta-\alpha)+\beta^2\gamma^2(\beta+\gamma)^2(\gamma-\beta)+\alpha^2\gamma^2(\alpha+\gamma)^2(\gamma-\alpha)\right]\]
Note that setting $\alpha=\gamma$ at this point would give zero, implying $\mu_{\gamma}=\mu_{\alpha}$, thereby proving the first part of the Lemma. Now considering the case where $\gamma\geqslant\beta$ and $\gamma>\alpha$, we can use the inequality $\beta^2\gamma^2(\beta+\gamma)^2(\gamma-\beta)\geqslant \beta^2\alpha^2(\beta+\alpha)^2(\gamma-\beta)$ to replace the second term and show that the previous line is greater than or equal to
\[\frac{\alpha+\beta+\gamma}{(\alpha+\gamma)^2(\beta+\gamma)^2(\alpha+\beta)^2}\left[\alpha^2\beta^2(\alpha+\beta)^2(\gamma-\alpha)+\alpha^2\gamma^2(\alpha+\gamma)^2(\gamma-\alpha)\right]\]which is strictly positive since $\gamma>\alpha$.
\end{proof}

There are then two cases we need to consider:
\begin{description}
\item[(i)] $\alpha,\beta,\gamma$ have a unique maximum, say $\gamma>\alpha,\beta$. 
\\Then by Lemma \ref{lem:mu}, $\mu_{\gamma}>\mu_{\alpha},\mu_{\beta}$, so $\alpha(\mu_{\gamma})$ and $\beta(\mu_{\gamma})$ are both positive and we can simulate $H(\mu_{\gamma})=\alpha(\mu_{\gamma}) XX +\beta(\mu_{\gamma})YY$ which is QMA-complete on the triangular lattice by Section \ref{sec:abtri}.
\item[(ii)] $\alpha<\beta=\gamma$. 
\\Then by Lemma \ref{lem:mu}, $\mu_{\alpha}<\mu_{\beta}=\mu_{\gamma}$, so we can simulate $H(\mu_{\gamma})=\alpha(\mu_{\gamma}) XX$. So for $\varepsilon >0$ small enough, $\mu=\mu_{\gamma} +\varepsilon$ satisfies $\alpha(\mu)>\beta(\mu),\gamma(\mu)$ and $\alpha(\mu),\beta(\mu),\gamma(\mu)$ all positive. So $H(\mu)$ is of the form we have just shown to be QMA-complete on the triangular lattice in (i).
\end{description}

This completes the proof of Theorem \ref{thm:triangle}.

\begin{reptheorem}{thm:triangle}
Let $H=\alpha XX+\beta YY + \gamma ZZ$ be a 2-qubit interaction such that $\alpha+\beta>0,$ $\alpha+\gamma>0,\beta+\gamma>0$ and $H$ is not proportional to $XX+YY+ZZ$. Then $\{H\}^+$\textsc{-Hamiltonian} is QMA-complete, even if the interactions are restricted to the edges of a 2D triangular lattice.
\end{reptheorem}

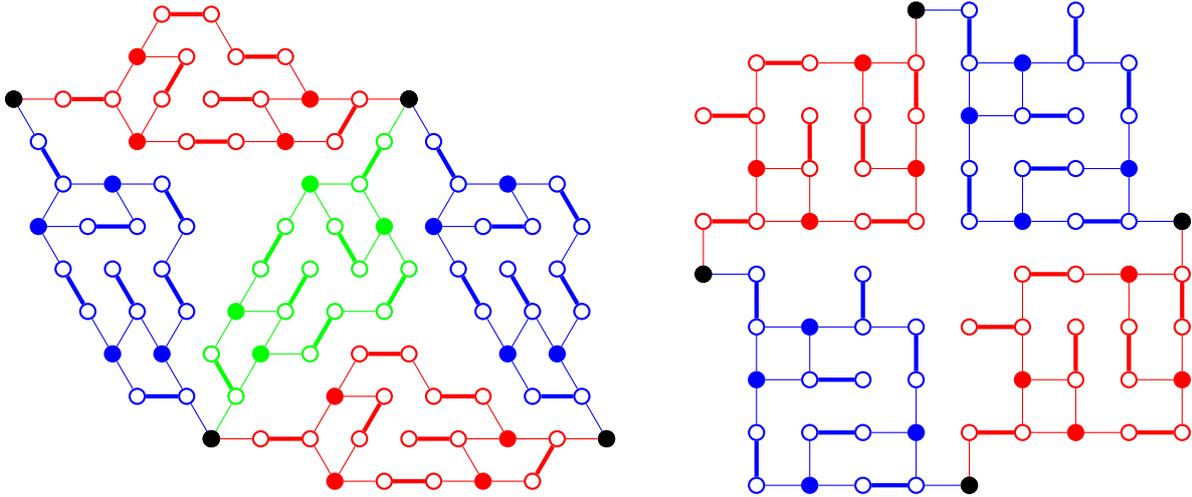
\begin{figure}
\centering
\begin{subfigure}{0.45\textwidth}
\centering
\begin{tikzpicture}[yscale=0.65,xscale=-0.65]

\begin{scope}[red]
\begin{scope}
\node[qubit,black] (1) at (0,0){};
\node[mqubit] (a1) at ($(1)+(1,0)$){};
\node[mqubit] (a2) at ($(a1)+(-60:1)$){};
\node[qubit] (2) at ($(a1)+(1,0)$){};
\node[qubit] (3) at ($(a2)+(1,0)$){};
\node[mqubit] (b1) at ($(2)+(1,0)$){};
\node[mqubit] (b2) at ($(b1)+(1,0)$){};

\draw (1) to (a1);
\draw (2) to (a1);
\draw (2) to (b1);
\draw (3) to (a2);
\draw (3) to (b1);
\draw[heavy] (a1) to (a2);
\draw[heavy] (b1) to (b2);

\end{scope}

\begin{scope}
\node[mqubit] (c2) at ($(b2)+(1,0)$){};
\node[mqubit] (c1) at ($(c2)+(120:1)$){};
\node[mqubit] (d1) at ($(c2)+(1,0)$){};
\node[mqubit] (d2) at ($(d1)+(1,0)$){};
\node[qubit] (4) at ($(c1)+(1,0)$){};
\node[qubit]  (5) at ($(c2)+(-60:1)$){};
\node[qubit,black] (6) at ($(d2)+(1,0)$){};

\draw (6) to (d2);
\draw (4) to (d1);
\draw (4) to (c1);
\draw (5) to (d1);
\draw (5) to (c2);
\draw[heavy] (c1) to (c2);
\draw[heavy] (d1) to (d2);
\end{scope}

\node[mqubit] (e1) at ($(2)+(60:1)$){};
\node[mqubit] (e2) at ($(e1)+(1,0)$){};
\node[mqubit] (e3) at ($(e2)+(60:1)$){};
\node[mqubit] (e4) at ($(e3)+(1,0)$){};
\draw (2) to (e1);
\draw (e3) to (e2);
\draw (e4) to (4);
\draw[heavy] (e1) to (e2);
\draw[heavy] (e3) to (e4);

\node[mqubit] (f1) at ($(3)+(1,0)$){};
\node[mqubit] (f2) at ($(f1)+(1,0)$){};
\draw (3) to (f1);
\draw (f2) to (5);
\draw[heavy] (f1) to (f2);
\end{scope}

\begin{scope}[blue,rotate=60]
\begin{scope}
\node[qubit,black] (1) at (0,0){};
\node[mqubit] (a1) at ($(1)+(1,0)$){};
\node[mqubit] (a2) at ($(a1)+(-60:1)$){};
\node[qubit] (2) at ($(a1)+(1,0)$){};
\node[qubit] (3) at ($(a2)+(1,0)$){};
\node[mqubit] (b1) at ($(2)+(1,0)$){};
\node[mqubit] (b2) at ($(b1)+(1,0)$){};

\draw (1) to (a1);
\draw (2) to (a1);
\draw (2) to (b1);
\draw (3) to (a2);
\draw (3) to (b1);
\draw[heavy] (a1) to (a2);
\draw[heavy] (b1) to (b2);

\end{scope}

\begin{scope}
\node[mqubit] (c2) at ($(b2)+(1,0)$){};
\node[mqubit] (c1) at ($(c2)+(120:1)$){};
\node[mqubit] (d1) at ($(c2)+(1,0)$){};
\node[mqubit] (d2) at ($(d1)+(1,0)$){};
\node[qubit] (4) at ($(c1)+(1,0)$){};
\node[qubit]  (5) at ($(c2)+(-60:1)$){};
\node[qubit,black] (6) at ($(d2)+(1,0)$){};

\draw (6) to (d2);
\draw (4) to (d1);
\draw (4) to (c1);
\draw (5) to (d1);
\draw (5) to (c2);
\draw[heavy] (c1) to (c2);
\draw[heavy] (d1) to (d2);
\end{scope}

\node[mqubit] (e1) at ($(2)+(60:1)$){};
\node[mqubit] (e2) at ($(e1)+(1,0)$){};
\node[mqubit] (e3) at ($(e2)+(60:1)$){};
\node[mqubit] (e4) at ($(e3)+(1,0)$){};
\draw (2) to (e1);
\draw (e3) to (e2);
\draw (e4) to (4);
\draw[heavy] (e1) to (e2);
\draw[heavy] (e3) to (e4);

\node[mqubit] (f1) at ($(3)+(1,0)$){};
\node[mqubit] (f2) at ($(f1)+(1,0)$){};
\draw (3) to (f1);
\draw (f2) to (5);
\draw[heavy] (f1) to (f2);
\end{scope}

\begin{scope}[green,shift={(8,0)},rotate=120]
\begin{scope}
\node[qubit,black] (1) at (0,0){};
\node[mqubit] (a1) at ($(1)+(1,0)$){};
\node[mqubit] (a2) at ($(a1)+(-60:1)$){};
\node[qubit] (2) at ($(a1)+(1,0)$){};
\node[qubit] (3) at ($(a2)+(1,0)$){};
\node[mqubit] (b1) at ($(2)+(1,0)$){};
\node[mqubit] (b2) at ($(b1)+(1,0)$){};

\draw (1) to (a1);
\draw (2) to (a1);
\draw (2) to (b1);
\draw (3) to (a2);
\draw (3) to (b1);
\draw[heavy] (a1) to (a2);
\draw[heavy] (b1) to (b2);

\end{scope}

\begin{scope}
\node[mqubit] (c2) at ($(b2)+(1,0)$){};
\node[mqubit] (c1) at ($(c2)+(120:1)$){};
\node[mqubit] (d1) at ($(c2)+(1,0)$){};
\node[mqubit] (d2) at ($(d1)+(1,0)$){};
\node[qubit] (4) at ($(c1)+(1,0)$){};
\node[qubit]  (5) at ($(c2)+(-60:1)$){};
\node[qubit,black] (6) at ($(d2)+(1,0)$){};

\draw (6) to (d2);
\draw (4) to (d1);
\draw (4) to (c1);
\draw (5) to (d1);
\draw (5) to (c2);
\draw[heavy] (c1) to (c2);
\draw[heavy] (d1) to (d2);
\end{scope}

\node[mqubit] (e1) at ($(2)+(60:1)$){};
\node[mqubit] (e2) at ($(e1)+(1,0)$){};
\node[mqubit] (e3) at ($(e2)+(60:1)$){};
\node[mqubit] (e4) at ($(e3)+(1,0)$){};
\draw (2) to (e1);
\draw (e3) to (e2);
\draw (e4) to (4);
\draw[heavy] (e1) to (e2);
\draw[heavy] (e3) to (e4);

\node[mqubit] (f1) at ($(3)+(1,0)$){};
\node[mqubit] (f2) at ($(f1)+(1,0)$){};
\draw (3) to (f1);
\draw (f2) to (5);
\draw[heavy] (f1) to (f2);
\end{scope}

\begin{scope}[red,shift={(4,4*sqrt(3))}]
\begin{scope}
\node[qubit,black] (1) at (0,0){};
\node[mqubit] (a1) at ($(1)+(1,0)$){};
\node[mqubit] (a2) at ($(a1)+(-60:1)$){};
\node[qubit] (2) at ($(a1)+(1,0)$){};
\node[qubit] (3) at ($(a2)+(1,0)$){};
\node[mqubit] (b1) at ($(2)+(1,0)$){};
\node[mqubit] (b2) at ($(b1)+(1,0)$){};

\draw (1) to (a1);
\draw (2) to (a1);
\draw (2) to (b1);
\draw (3) to (a2);
\draw (3) to (b1);
\draw[heavy] (a1) to (a2);
\draw[heavy] (b1) to (b2);

\end{scope}

\begin{scope}
\node[mqubit] (c2) at ($(b2)+(1,0)$){};
\node[mqubit] (c1) at ($(c2)+(120:1)$){};
\node[mqubit] (d1) at ($(c2)+(1,0)$){};
\node[mqubit] (d2) at ($(d1)+(1,0)$){};
\node[qubit] (4) at ($(c1)+(1,0)$){};
\node[qubit]  (5) at ($(c2)+(-60:1)$){};
\node[qubit,black] (6) at ($(d2)+(1,0)$){};

\draw (6) to (d2);
\draw (4) to (d1);
\draw (4) to (c1);
\draw (5) to (d1);
\draw (5) to (c2);
\draw[heavy] (c1) to (c2);
\draw[heavy] (d1) to (d2);
\end{scope}

\node[mqubit] (e1) at ($(2)+(60:1)$){};
\node[mqubit] (e2) at ($(e1)+(1,0)$){};
\node[mqubit] (e3) at ($(e2)+(60:1)$){};
\node[mqubit] (e4) at ($(e3)+(1,0)$){};
\draw (2) to (e1);
\draw (e3) to (e2);
\draw (e4) to (4);
\draw[heavy] (e1) to (e2);
\draw[heavy] (e3) to (e4);

\node[mqubit] (f1) at ($(3)+(1,0)$){};
\node[mqubit] (f2) at ($(f1)+(1,0)$){};
\draw (3) to (f1);
\draw (f2) to (5);
\draw[heavy] (f1) to (f2);
\end{scope}

\begin{scope}[blue,shift={(8,0)},rotate=60]
\begin{scope}
\node[qubit,black] (1) at (0,0){};
\node[mqubit] (a1) at ($(1)+(1,0)$){};
\node[mqubit] (a2) at ($(a1)+(-60:1)$){};
\node[qubit] (2) at ($(a1)+(1,0)$){};
\node[qubit] (3) at ($(a2)+(1,0)$){};
\node[mqubit] (b1) at ($(2)+(1,0)$){};
\node[mqubit] (b2) at ($(b1)+(1,0)$){};

\draw (1) to (a1);
\draw (2) to (a1);
\draw (2) to (b1);
\draw (3) to (a2);
\draw (3) to (b1);
\draw[heavy] (a1) to (a2);
\draw[heavy] (b1) to (b2);

\end{scope}

\begin{scope}
\node[mqubit] (c2) at ($(b2)+(1,0)$){};
\node[mqubit] (c1) at ($(c2)+(120:1)$){};
\node[mqubit] (d1) at ($(c2)+(1,0)$){};
\node[mqubit] (d2) at ($(d1)+(1,0)$){};
\node[qubit] (4) at ($(c1)+(1,0)$){};
\node[qubit]  (5) at ($(c2)+(-60:1)$){};
\node[qubit,black] (6) at ($(d2)+(1,0)$){};

\draw (6) to (d2);
\draw (4) to (d1);
\draw (4) to (c1);
\draw (5) to (d1);
\draw (5) to (c2);
\draw[heavy] (c1) to (c2);
\draw[heavy] (d1) to (d2);
\end{scope}

\node[mqubit] (e1) at ($(2)+(60:1)$){};
\node[mqubit] (e2) at ($(e1)+(1,0)$){};
\node[mqubit] (e3) at ($(e2)+(60:1)$){};
\node[mqubit] (e4) at ($(e3)+(1,0)$){};
\draw (2) to (e1);
\draw (e3) to (e2);
\draw (e4) to (4);
\draw[heavy] (e1) to (e2);
\draw[heavy] (e3) to (e4);

\node[mqubit] (f1) at ($(3)+(1,0)$){};
\node[mqubit] (f2) at ($(f1)+(1,0)$){};
\draw (3) to (f1);
\draw (f2) to (5);
\draw[heavy] (f1) to (f2);
\end{scope}

\end{tikzpicture}
\end{subfigure}
\hfill
\begin{subfigure}{0.45\textwidth}
\centering
\begin{tikzpicture}[scale=0.7]
\begin{scope}[red,rotate=90,shift={(-1,-4)}]
\begin{scope}[red]
\node[qubit,black] (1) at (1,4){};
\node[qubit] (2) at (3,3){};
\node[qubit] (3) at (2,2){};

\node[mqubit] (a1) at (2,4){};
\node[mqubit] (a2) at (2,3){};
\draw [heavy] (a1) to (a2);
\node[mqubit] (b1) at (3,2){};
\node[mqubit] (b2) at (4,2){};
\draw [heavy] (b1) to (b2);

\draw (1) to (a1);
\draw (2) to (a2);
\draw (3) to (a2);
\draw (2) to (b1);
\draw (3) to (b1);
\end{scope}

\begin{scope}
\node[qubit] (4) at (5,1){};
\node[qubit] (5) at (3,0){};
\node[qubit,black] (6) at (6,0){};

\node[mqubit] (c1) at (3,1){};
\node[mqubit] (c2) at (4,1){};
\draw [heavy] (c1) to (c2);
\node[mqubit] (d1) at (4,0){};
\node[mqubit] (d2) at (5,0){};
\draw [heavy] (d1) to (d2);

\draw (4) to (c2);
\draw (4) to (d2);
\draw (5) to (c1);
\draw (5) to (d1);
\draw (6) to (d2);
\end{scope}

\node[mqubit] (e1) at (4,4){};
\node[mqubit] (e2) at (4,3){};
\node[mqubit] (e3) at (5,3){};
\node[mqubit] (e4) at (5,2){};
\draw [heavy] (e1) to (e2);
\draw [heavy] (e3) to (e4);
\draw (2) to (e2);
\draw (e2) to (e3);
\draw (e4) to (4);

\node[mqubit] (f1) at (2,1){};
\node[mqubit] (f2) at (2,0){};

\draw [heavy] (f1) to (f2);
\draw (3) to (f1);
\draw (f2) to (5);
\end{scope}

\begin{scope}[blue,shift={(-1,-4)}]
\begin{scope}
\node[qubit,black] (1) at (1,4){};
\node[qubit] (2) at (3,3){};
\node[qubit] (3) at (2,2){};

\node[mqubit] (a1) at (2,4){};
\node[mqubit] (a2) at (2,3){};
\draw [heavy] (a1) to (a2);
\node[mqubit] (b1) at (3,2){};
\node[mqubit] (b2) at (4,2){};
\draw [heavy] (b1) to (b2);

\draw (1) to (a1);
\draw (2) to (a2);
\draw (3) to (a2);
\draw (2) to (b1);
\draw (3) to (b1);
\end{scope}

\begin{scope}
\node[qubit] (4) at (5,1){};
\node[qubit] (5) at (3,0){};
\node[qubit,black] (6) at (6,0){};

\node[mqubit] (c1) at (3,1){};
\node[mqubit] (c2) at (4,1){};
\draw [heavy] (c1) to (c2);
\node[mqubit] (d1) at (4,0){};
\node[mqubit] (d2) at (5,0){};
\draw [heavy] (d1) to (d2);

\draw (4) to (c2);
\draw (4) to (d2);
\draw (5) to (c1);
\draw (5) to (d1);
\draw (6) to (d2);
\end{scope}

\node[mqubit] (e1) at (4,4){};
\node[mqubit] (e2) at (4,3){};
\node[mqubit] (e3) at (5,3){};
\node[mqubit] (e4) at (5,2){};
\draw [heavy] (e1) to (e2);
\draw [heavy] (e3) to (e4);
\draw (2) to (e2);
\draw (e2) to (e3);
\draw (e4) to (4);

\node[mqubit] (f1) at (2,1){};
\node[mqubit] (f2) at (2,0){};

\draw [heavy] (f1) to (f2);
\draw (3) to (f1);
\draw (f2) to (5);
\end{scope}

\begin{scope}[red,shift={(5,-4)},rotate=90,shift={(-1,-4)}]
\begin{scope}[red]
\node[qubit,black] (1) at (1,4){};
\node[qubit] (2) at (3,3){};
\node[qubit] (3) at (2,2){};

\node[mqubit] (a1) at (2,4){};
\node[mqubit] (a2) at (2,3){};
\draw [heavy] (a1) to (a2);
\node[mqubit] (b1) at (3,2){};
\node[mqubit] (b2) at (4,2){};
\draw [heavy] (b1) to (b2);

\draw (1) to (a1);
\draw (2) to (a2);
\draw (3) to (a2);
\draw (2) to (b1);
\draw (3) to (b1);
\end{scope}

\begin{scope}
\node[qubit] (4) at (5,1){};
\node[qubit] (5) at (3,0){};
\node[qubit,black] (6) at (6,0){};

\node[mqubit] (c1) at (3,1){};
\node[mqubit] (c2) at (4,1){};
\draw [heavy] (c1) to (c2);
\node[mqubit] (d1) at (4,0){};
\node[mqubit] (d2) at (5,0){};
\draw [heavy] (d1) to (d2);

\draw (4) to (c2);
\draw (4) to (d2);
\draw (5) to (c1);
\draw (5) to (d1);
\draw (6) to (d2);
\end{scope}

\node[mqubit] (e1) at (4,4){};
\node[mqubit] (e2) at (4,3){};
\node[mqubit] (e3) at (5,3){};
\node[mqubit] (e4) at (5,2){};
\draw [heavy] (e1) to (e2);
\draw [heavy] (e3) to (e4);
\draw (2) to (e2);
\draw (e2) to (e3);
\draw (e4) to (4);

\node[mqubit] (f1) at (2,1){};
\node[mqubit] (f2) at (2,0){};

\draw [heavy] (f1) to (f2);
\draw (3) to (f1);
\draw (f2) to (5);
\end{scope}

\begin{scope}[blue,shift={(3,1)}]
\begin{scope}
\node[qubit,black] (1) at (1,4){};
\node[qubit] (2) at (3,3){};
\node[qubit] (3) at (2,2){};

\node[mqubit] (a1) at (2,4){};
\node[mqubit] (a2) at (2,3){};
\draw [heavy] (a1) to (a2);
\node[mqubit] (b1) at (3,2){};
\node[mqubit] (b2) at (4,2){};
\draw [heavy] (b1) to (b2);

\draw (1) to (a1);
\draw (2) to (a2);
\draw (3) to (a2);
\draw (2) to (b1);
\draw (3) to (b1);
\end{scope}

\begin{scope}
\node[qubit] (4) at (5,1){};
\node[qubit] (5) at (3,0){};
\node[qubit,black] (6) at (6,0){};

\node[mqubit] (c1) at (3,1){};
\node[mqubit] (c2) at (4,1){};
\draw [heavy] (c1) to (c2);
\node[mqubit] (d1) at (4,0){};
\node[mqubit] (d2) at (5,0){};
\draw [heavy] (d1) to (d2);

\draw (4) to (c2);
\draw (4) to (d2);
\draw (5) to (c1);
\draw (5) to (d1);
\draw (6) to (d2);
\end{scope}

\node[mqubit] (e1) at (4,4){};
\node[mqubit] (e2) at (4,3){};
\node[mqubit] (e3) at (5,3){};
\node[mqubit] (e4) at (5,2){};
\draw [heavy] (e1) to (e2);
\draw [heavy] (e3) to (e4);
\draw (2) to (e2);
\draw (e2) to (e3);
\draw (e4) to (4);

\node[mqubit] (f1) at (2,1){};
\node[mqubit] (f2) at (2,0){};

\draw [heavy] (f1) to (f2);
\draw (3) to (f1);
\draw (f2) to (5);
\end{scope}

\end{tikzpicture}
\end{subfigure}
\caption{Tessellations of the square and triangular lattice gadgets.}
\label{fig:tessel}
\end{figure}

\subsection{2D square lattice}

A 2D square lattice is a bipartite graph, and for any bipartite graph we can locally change basis for every qubit in one side of the partition by conjugating by $Z$. Restricted to positive weights, this will take any QMA-complete interaction to one contained in StoqMA, by effectively taking $\alpha,\beta \rightarrow -\alpha,-\beta$. Since we do not expect QMA=StoqMA, we will instead study the case where both positive and negative weights are allowed.

To show that the XY interaction is QMA-complete on a 2D square lattice, we proceed as with the triangular lattice. We place a fine square grid over the 2D planar interaction graph of Section \ref{sec:planar}, move vertices to the nearest lattice points and deform edges to paths on the lattice (re-routing in a small region around each vertex if necessary to avoid collisions). We then use the subdivision gadget (Section \ref{sec:subdivision}) to implement effective interactions across the ends of the paths. However, on a bipartite graph such as the square lattice, the parity of the length of any path between two points is fixed. Therefore, if there are an odd number of qubits on the path but we wish to simulate a positive interaction (or vice versa), we need one of the interactions along this path to be negative. This is then enough to show that the XY interaction is QMA-complete on a 2D square lattice with positive and negative weights.

\begin{figure}
\centering
\begin{tikzpicture}[scale=0.8]

\begin{scope}[blue]
\node[qubit,black] (1) at (1,4)[label={[black]left:1}]{};
\node[qubit,black] (2) at (3,3)[label={[black]above:2}]{};
\node[qubit,black] (3) at (2,2)[label={[black]left:3}]{};

\node[mqubit] (a1) at (2,4){};
\node[mqubit] (a2) at (2,3){};
\draw [heavy] (a1) to (a2);
\node[mqubit] (b1) at (3,2){};
\node[mqubit] (b2) at (4,2){};
\draw [heavy] (b1) to (b2);

\draw (1) to (a1);
\draw (2) to (a2);
\draw (3) to (a2);
\draw[dashed] (2) to (b1);
\draw (3) to (b1);
\end{scope}

\begin{scope}[red]
\node[qubit,black] (4) at (5,1)[label={[black]right:4}]{};
\node[qubit,black] (5) at (3,0)[label={[black]below:5}]{};
\node[qubit,black] (6) at (6,0)[label={[black]right:6}]{};

\node[mqubit] (c1) at (3,1){};
\node[mqubit] (c2) at (4,1){};
\draw [heavy] (c1) to (c2);
\node[mqubit] (d1) at (4,0){};
\node[mqubit] (d2) at (5,0){};
\draw [heavy] (d1) to (d2);

\draw[dashed] (4) to (c2);
\draw (4) to (d2);
\draw (5) to (c1);
\draw (5) to (d1);
\draw (6) to (d2);
\end{scope}

\begin{scope}[green]
\node[mqubit] (e1) at (4,4){};
\node[mqubit] (e2) at (4,3){};
\node[mqubit] (e3) at (5,3){};
\node[mqubit] (e4) at (5,2){};
\draw [heavy] (e1) to (e2);
\draw [heavy] (e3) to (e4);
\draw[dashed] (2) to (e2);
\draw (e2) to (e3);
\draw (e4) to (4);

\node[mqubit] (f1) at (2,1){};
\node[mqubit] (f2) at (2,0){};

\draw [heavy] (f1) to (f2);
\draw (3) to (f1);
\draw (f2) to (5);
\end{scope}
\end{tikzpicture}
\caption{Square lattice gadget for simulating $XX+YY$. Dashed lines represent negative interactions.}
\label{fig:square}
\end{figure}

For a more general symmetric interaction $H=\alpha XX+\beta YY +\gamma ZZ$, given that we are allowed both positive and negative weights and can conjugate one side of the partition of the graph by $X,Y$ or $Z$, we can assume without loss of generality that $\alpha,\beta,\gamma$ are all non-negative. Figure \ref{fig:square} shows a gadget that fits on the 2D square lattice and simulates the interaction $-\lambda H^{(1)}+\mu H^{(2)}$ in exactly the same way as the gadget on the 
triangular lattice, by first simulating the six qubit Hamiltonian shown in Figure \ref{fig:triangle}(b). Again the blue and red gadgets act as fork gadgets, but in order for the unwanted $\widetilde{H}_{23}$ and $\widetilde{H}_{56}$ interactions to be correctly cancelled out, the sign of the interactions marked as dashed lines in Figure \ref{fig:square} must be negative. The other dashed line in Figure \ref{fig:square} must also be negative so that the gadget between qubits 2 and 3 correctly simulates $+H'_{23}$. It is also shown in Figure \ref{fig:tessel} that this gadget will tessellate on a square lattice without overlapping.

Therefore, using the prescription set out for the triangular lattice in the previous section, we can use $-\lambda H^{(1)}+\mu H^{(2)}$ to generate $XX+YY$.
We have thus shown that the interaction $\alpha XX + \beta YY +\gamma ZZ$ is QMA-complete on a square lattice with positive and negative weights, as long as two of $\alpha,\beta, \gamma$ are non-zero, and they are not all equal.

\subsubsection{Interactions with non-trivial 1-local part}
We will finally consider the more general case of 2-qubit interactions $H^{\text{full}}$ with a non-trivial 1-local part. It is sufficient to assume that $H^{\text{full}}$ is symmetric or antisymmetric under interchange of the two qubits~\cite{cubitt14}. For a symmetric Hamiltonian we can further assume that
\[H^{\text{full}}=\alpha XX+\beta YY+\gamma ZZ +AI+IA=H+AI+IA\]
where $A$ is a general 1-local term. An antisymmetric interaction has, up to normalization, the normal form $XZ-ZX +AI-IA$.
 
In both symmetric and antisymmetric cases, it is shown in~\cite{cubitt14} that the Hamiltonian 
\[H_0=H^{\text{full}}_{ab}-H^{\text{full}}_{cb}+H^{\text{full}}_{cd}-H^{\text{full}}_{ad}=H_{ab}-H_{cb}+H_{cd}-H_{ad}\] has a unique ground state $|\Omega\>$, and that in this state $|\Omega\>$ the reduced density matrix for each of the qubits $a,b,c,d$ is $I/2$. 

Therefore if we  project into the ground state of the Hamiltonian $H_0$, and apply an extra interaction $V=-H^{\text{full}}_{ed}$ for example, then the first order perturbation term will be
\[V_{--}=\<\Omega|-H^{\text{full}}_{ed}|\Omega\>=-A_{e}.\]
So by Lemma \ref{lem:1st}, we can simulate any multiple of $-A_{e}$, which we could use to cancel out the 1-local part of the interaction $H^{\text{full}}$ acting on $e$.

\begin{figure}
\centering
\begin{tikzpicture}
\foreach \i in {0,...,2}
	{\node[qubit] (\i,0) at (\i,0){};
	\node[qubit] (0,\i+1) at (0,\i+1){};
	\node[qubit] (\i+1,3) at (\i+1,3){};
	\node[qubit] (3,\i) at (3,\i){};}

\node[mqubit] (a) at (1,1){};
\node[mqubit] (b) at (1,2){};
\node[mqubit] (c) at (2,1){};
\node[mqubit] (d) at (2,2){};

\draw [heavy] (a) to (b);
\draw [heavy] (b) to (d);
\draw [heavy] (d) to (c);
\draw [heavy] (c) to (a);

\draw  (0,1) to (a);
\draw  (0,2) to (b);
\draw  (1,0) to (a);
\draw  (2,0) to (c);
\draw  (3,1) to (c);
\draw  (3,2) to (d);
\draw  (1,3) to (b);
\draw  (2,3) to (d);
\end{tikzpicture}
\caption{Square lattice gadget for removing 1-local terms}
\label{fig:1local}
\end{figure}
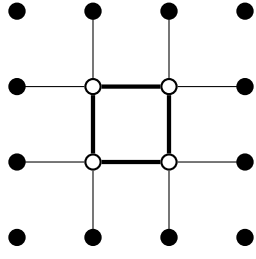

This gadget can be implemented on a square lattice as shown in Figure \ref{fig:1local}, to simulate arbitrary multiples of the 1-local term $A$ on each of the qubits connected to the central square.

These gadgets can be applied in parallel over a square lattice divided into $3 \times 3$ blocks, with additional physical $H^{\text{full}}$ interactions along the outside of each block, to simulate a Hamiltonian of the form depicted in Figure \ref{fig:1locallattice}. All of the physical interactions shown on this interaction graph are of the form $H^{\text{full}}$, but the 1-local terms simulated using these gadgets are chosen such that they exactly cancel out the 1-local part of $H^{\text{full}}$ on all the white qubits.

We label each black qubit with a label $i$, and label the pair of white qubits between $i$ and $j$ as $a_{ij}$ and $a_{ji}$, where $a_{ij}$ is the qubit nearest to $i$. We will construct a Hamiltonian acting on the interaction graph of Figure \ref{fig:1locallattice} that simulates a general Hamiltonian of the form $\sum \lambda_{ij} \widetilde{H}_{ij}$, where interactions take place between adjacent vertices on a square lattice, which was shown to be QMA-complete in the previous section.

\begin{figure}[b]
\centering
\begin{tikzpicture}
\foreach \x in {0,1,2,3}
\foreach \y in {0,1,2}
{
\begin{scope}[scale=0.5,shift={(3*\x,3*\y)}]
\node[qubit] (1) at (0,0){};
\node[qubit] (2) at (0,3){};
\node[qubit] (3) at (3,0){};
\node[qubit] (4) at (3,3){};

\node[mqubit] (12) at (0,1){};
\node[mqubit] (21) at (0,2){};
\draw[heavy] (12) to (21);
\draw (1) to (12);
\draw (2) to (21);

\node[mqubit] (13) at (1,0){};
\node[mqubit] (31) at (2,0){};
\draw[heavy] (13) to (31);
\draw (1) to (13);
\draw (3) to (31);

\node[mqubit] (24) at (1,3){};
\node[mqubit] (42) at (2,3){};
\draw[heavy] (24) to (42);
\draw (2) to (24);
\draw (4) to (42);

\node[mqubit] (34) at (3,1){};
\node[mqubit] (43) at (3,2){};
\draw[heavy] (34) to (43);
\draw (3) to (34);
\draw (4) to (43);

\end{scope}}
\end{tikzpicture}
\caption{Hamiltonian simulated by parallel use of gadget}
\label{fig:1locallattice}
\end{figure}
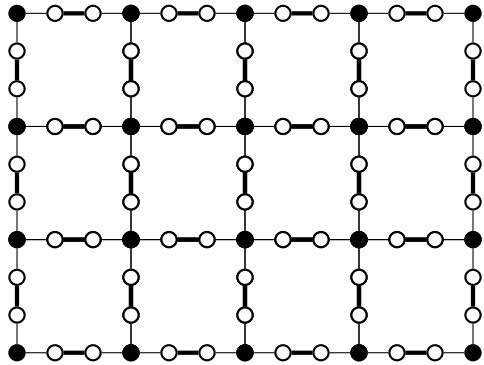

On each pair of white qubits we apply a heavy 2-local $H$ interaction, so that the overall heavily weighted term is
\[H_0=\sum_{(i,j) \in E} \frac{J}{|\lambda_{ij}|} H_{a_{ij} a_{ji}}\]
where $E$ is the set of edges of the square lattice and $J$ is a large weight to be defined.

We choose the weights of the other interactions so that the 1-local terms all cancel out on the black qubits, and then we can simulate $\widetilde{H}$ interactions as usual by applying the basic gadget to the resulting $H$ interactions. 
In order to do this we pick all the other weights (those between one black and one white qubit) to be of the same magnitude $\sqrt{J}$, where $J=O(\text{poly}(n))$ is the max size of the weights in the target Hamiltonian, and let the sign of the interaction between qubits $i$ and $a_{ij}$ be given by $\mu_{ij}=\pm 1$, so overall we have
\[V=\sqrt{J}\sum_{(i,j) \in E} \mu_{ij}(H_{ia_{ij}}+A_{i})+\mu_{ji}(H_{ja_{ji}}+A_{j}).\]
Therefore, ensuring that the 1-local terms cancel out on the $i$th black qubit is equivalent to choosing $\mu_{ij}$ such that $\sum_j \mu_{ij}=0$. In order to simulate the target Hamiltonian $H_{\text{target}}=\sum\lambda_{ij}\widetilde{H}_{ij}$, we need $\mu_{ij}\mu_{ji}=\text{sgn}(\lambda_{ij})$.

Suppose the black qubits along the top row are labelled $1,2,\dots$ from left to right. Pick (arbitrarily) $\mu_{12}=+1$, and then choose inductively all the $\mu_{ij}$ values along the top row according to the following rules:
\[\mu_{i+1,i}=\text{sgn}(\lambda_{i,i+1})\mu_{i,i+1}\quad \text{and} \quad \mu_{i,i+1}=-\mu_{i,i-1}.\] Then do the same for every other row as well, and repeat this procedure for each column. Then only the black qubits on the edge of the lattice still have 1-local terms, and these can be removed simply by adding extra square gadgets to the outside of the lattice. 

Therefore in the symmetric case, where $H=\alpha XX +\beta YY+\gamma ZZ$, the overall Hamiltonian $\Delta H_0+\Delta^{1/2}V$ will simulate the effective Hamiltonian
\[H_{\text{eff}}=-V_{-+}H_0^{-1}V_{+-}=\sum_{(i,j) \in E}\lambda_{ij}\widetilde{H}_{ij}, \]
where $\widetilde{H}$ is defined as in eqn.\ (\ref{eq:htilde}). Thus we can simulate any Hamiltonian consisting of $\widetilde{H}$ terms on a square lattice. In the antisymmetric case, $H=XZ-ZX$, we can choose signs in exactly the same way so that all 1-local terms are cancelled out and we simulate the target Hamiltonian 
\[H_{\text{target}}=\sum_{(i,j)\in E}\lambda_{ij} \left(X_iZ_j-Z_iX_j\right). \]

It remains to show that this Hamiltonian is QMA-complete, even when restricted to a 2D square lattice. We achieve this by simulating $XX+ZZ$ terms on the lattice. When two qubits are connected to the same qubit of a mediator pair using $XZ-ZX$ interactions, as in Figure \ref{fig:antisym1}, we showed in Section \ref{sec:antisym} that an effective $XX+ZZ$ is simulated. However, this gadget does not tessellate on the lattice, as there is not enough space to fit the two mediator qubits between physical qubits. In Section \ref{sec:antisym}, we also showed that two qubits connected to opposite ends of a mediator qubit pair can simulate an effective $XZ-ZX$ interaction, so this is a kind of subdivision gadget. If one of the edges of the gadget in Figure \ref{fig:antisym1} is simulated using this subdivision gadget, we obtain the 6 qubit gadget in Figure \ref{fig:antisym2}. This 6 qubit gadget will tessellate on a 2D square lattice, so we can use it to simulate the XY interaction (after relabelling $Z$ to $Y$) on a 2D square lattice.

\begin{figure}[t]
\centering
\begin{subfigure}{0.4\textwidth}
\centering
\begin{tikzpicture}
\node[qubit] (1) at (0,0){};
\node[mqubit] (2) at (1,0){};
\node[mqubit] (3) at (1,1){};
\node[qubit] (4) at (2,0){};

\draw[heavy] (2) to (3);
\draw (1) to (2);
\draw (2) to (4);
\end{tikzpicture}
\caption{Basic gadget}
\label{fig:antisym1}
\end{subfigure}
\quad
\begin{subfigure}{0.4\textwidth}
\centering
\begin{tikzpicture}
\node[qubit] (1) at (0,0){};
\node[qubit] (2) at (1,0){};
\node[qubit] (3) at (1,1){};
\node[mqubit] (4) at (2,0){};
\node[mqubit] (5) at (3,0){};
\node[qubit] (6) at (4,0){};

\draw[heavy] (4) to (5);
\draw (1) to (2);
\draw (3) to (2);
\draw (2) to (4);
\draw (5) to (6);
\end{tikzpicture}
\caption{6 qubit gadget}
\label{fig:antisym2}
\end{subfigure}
\caption{Gadgets for simulating $XX+ZZ$ using $XZ-ZX$ interactions. The 6 qubit gadget simulates the basic gadget to implement an effective $XX+ZZ$ interaction across the left- and rightmost qubits.}
\end{figure}
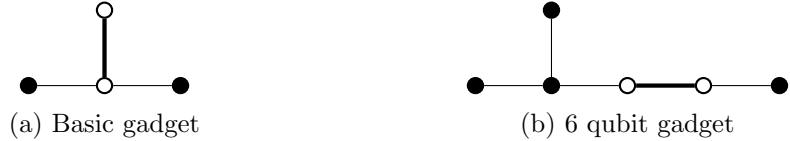

We have finally completed the proof of Theorem \ref{thm:square}.

\begin{reptheorem}{thm:square}
Let $H$ be a 2-local qubit interaction with Pauli rank at least 2, such that the 2-local part of $H$ is not proportional to $XX+YY+ZZ$. Then $\{H\}$\textsc{-Hamiltonian} is QMA-complete, even when the interactions are restricted to a 2D square lattice.
\end{reptheorem}

\section{The XY model on a cyclic chain}
\label{sec:XYcyclic}
This section contains the deferred proof of Lemma \ref{lem:xy}. Recall that in this lemma we consider the XY model on a cyclic quantum spin chain of $N$ qubits, for $N$ even but not a multiple of 4. This is a system of $N$ qubits on a circle with an interaction of the form $XX+YY$ between nearest neighbour qubits, such that the overall Hamiltonian is 
\[H=\frac{1}{4}\left[\sum_{i=1}^{N-1}(X_{i}X_{i+1}+Y_{i}Y_{i+1}) \: + X_1X_N+Y_1Y_N\right]\] 
The overall factor of $\frac{1}{4}$ is included in order to match the original formulation in~\cite{Lieb1961} and the notation used in the rest of literature.
The properties of the system that we are most interested in are the spectral gap (the energy difference between the ground state $|\Omega\>$ and the first excited state) and the spin correlation functions $\rho=\<\Omega|X_iX_j+Y_iY_j|\Omega\>$ (the expectation value of $XX+YY$ in the ground state). In particular we will show that both of these quantities are $\Omega(1/\text{poly}(N))$.

This model has been extensively studied since the seminal paper by Lieb, Schultz and Mattis~\cite{Lieb1961}. For a helpful review of the area see~\cite{Metlitski2004}. However, works in this area usually consider the infinite chain, whereas we will need to take $N$ finite (but growing). In particular, asymptotic expressions for the spin correlation functions of the infinite chain were determined in~\cite{McCoy68}. Other related work studying correlation functions in varying settings (finite temperature, anisotropic interactions, etc.)\ includes~\cite{McCoy71,perk83}. These can be approached using a mapping to the 2D Ising model (see e.g.~\cite{perk77}). Here we will start by considering the infinite chain, and then show that the finite chain provides an asymptotically precise approximation. Following the exposition in~\cite{Metlitski2004}, we provide self-contained proofs of the results that we need.

\subsection{Ground State and Spectral Gap}

As is shown in~\cite{Lieb1961} the Hamiltonian $H$ can be transformed to a system of free fermions with nearest neighbour interactions in the following way:

First let \[a_j=\frac{1}{2}(X_j+iY_j) \quad \text{and} \quad a_j^{\dagger}=\frac{1}{2}(X_j-iY_j) \quad \text{for } j=1,...,n\] so that the Hamiltonian of the system is now:
\[H=\frac{1}{2}\left[\sum_{i=1}^{N-1}(a_{i}^{\dagger}a_{i+1}+a_{i}a_{i+1}^{\dagger}) \: + a_1^{\dagger}a_N+a_1a_N^{\dagger}\right]\]

Note that the $a_i$ do not satisfy canonical commutation relations for a boson or canonical anticommutation relations of a fermion, but instead a mix of the two: $\{a_i,a_i^{\dagger}\}=I,\{a_i,a_i\}=0$ and $[a_i^{\dagger},a_j]=0=[a_i,a_j]$ for $i\neq j$ . So we apply a Jordan-Wigner transformation which takes the Hamiltonian into the form of one for a system of fermions, which can then be solved simply. Let
\[c_i=\exp\left(i\pi\sum_{j=1}^{i-1} a_j^{\dagger}a_j\right)a_i \quad \text{and} \quad c_i^{\dagger}=a_i^{\dagger}\exp\left(-i\pi\sum_{j=1}^{i-1} a_j^{\dagger}a_j\right)\]
which satisfy the canonical fermionic anticommutation relations $\{c_i,c_j^\dagger\}=\delta_{ij}$ and $\{c_i,c_j\}=0$. For more information on these operators see for example~\cite{Metlitski2004}. In terms of the $c_i$'s the Hamiltonian becomes 
\[H=\frac{1}{2}\left[\sum_{i=1}^{N-1}(c_{i}^{\dagger}c_{i+1}+c_{i}c^{\dagger}_{i+1})-(c_{N}^{\dagger}c_{1}+c_{1}^{\dagger}c_{N})\exp(i\pi \sum_{j=1}^{N} c_j^{\dagger}c_j)\right]\]

At this point Lieb, Schultz and Mattis~\cite{Lieb1961} turn to considering the similar Hamiltonian
\[H_- = \frac{1}{2}\left[\sum_{i=1}^{N-1}(c_{i}^{\dagger}c_{i+1}+c_{i}c^{\dagger}_{i+1})+c_{N}^{\dagger}c_{1}+c_{1}^{\dagger}c_{N}\right]\] 
as they are interested only in the leading order asymptotics, and this change only results in an $O(1/N)$ error in the calculation of the eigenvalues. However, we will need to be more careful here.

Let $P=\exp(i\pi \mathcal{N})$ where $\mathcal{N}=\sum_{j=1}^{N}c_j^{\dagger}c_j$ is the fermionic number operator which has eigenvalues $0,1,...,N$, so the eigenvalues of $P$ are $\pm 1$. Note that $P$ commutes with $H$, so we can diagonalise $H$ into block diagonal form and consider its action on the two eigenspaces of $P$ separately. Then

\[H=H_-\frac{1-P}{2}+H_+\frac{1+P}{2} \quad \text{where } H_- \text{ is as above, and }\]
\[H_+=\frac{1}{2}\left[\sum_{i=1}^{N-1}(c_{i}^{\dagger}c_{i+1}+c_{i}c^{\dagger}_{i+1})-(c_{N}^{\dagger}c_{1}+c_{1}^{\dagger}c_{N})\right]\] 

Now we can diagonalise both $H_-$ and $H_+$ separately, by first writing $H_-$ as $\sum c_i^{\dagger}A_{ij}c_{j}$, and finding the unitary matrix $U$ that diagonalises $A$. 

\[H_-=\sum c_i^{\dagger}A_{ij}c_{j} \quad \Rightarrow \quad A_{ij}=\left\lbrace \begin{array}{lc}
(\delta_{2j}+\delta_{jN})/2 &  i=1 \\
(\delta_{i+1 \: j}+\delta_{j\: i-1})/2 &  1<i<N \\
(\delta_{1j}+\delta_{j\: N-1})/2 &  i=N \\
\end{array}
\right.\]
Let $\phi_{kj}$ denote the $j$th coordinate of the $k$th eigenvector of $A$. We can easily verify that with $\phi_{kj}$ as below, $\{\phi_k\}$ is a complete orthonormal set of eigenvectors with corresponding eigenvalues $\Lambda_k= \cos(2\pi k/N)$ :
\[\phi_{kj}=\frac{1}{\sqrt{N}}e^{2\pi i kj/N}, \quad \text{ where }  0\leqslant k \leqslant N-1\]

 Then by making a second change of variables $\eta_k=\sum_{i}\phi^{\star}_{ki}c_j$ (or conversely $c_i=\sum_{k}\phi_{ki}\eta_k$), which are also fermionic operators, the Hamiltonian becomes:
 
 \[H_-=\sum_k \Lambda_k\eta_k^{\dagger}\eta_k\]
 To make the ground state energy clear, we will make one final change of variables 
 \[\xi_k=\left\lbrace \begin{array}{cc}
 \eta_k & \Lambda_k \geqslant 0 \\
 \eta_k^{\dagger} & \Lambda_k < 0 \\
\end{array}  \right. \]
Then the Hamiltonian $H_-$ becomes 
\[H_-= \sum_{k:\Lambda_k \geqslant 0} \Lambda_k \xi_k^{\dagger} \xi_k + \sum_{k:\Lambda_k < 0} \Lambda_k \xi_k \xi_k^{\dagger}=\sum_{k} |\Lambda_k| \xi_k^{\dagger} \xi_k +\left(\sum_{k:\Lambda_k < 0} \Lambda_k\right)I\]
 And so the ground state of $H_-$ is $|\Omega\>$, the unique state that satisfies $\xi_k |\Omega\>=0$ for all $k$, and which has energy equal to 
\[\sum_{k:\Lambda_k<0} \Lambda_k=\sum_{n=\frac{N+2}{4}}^{\frac{3N-2}{4}} \cos \left(\frac{2\pi n}{N}\right)=\Re\left[e^{\frac{2\pi i(N+2)}{4N}}\sum_{n=0}^{\frac{N}{2}-1} e^{\frac{2\pi in}{N}}\right]\]
\[=\Re \left[i e^{\frac{i\pi}{N}}\frac{2}{1-e^{\frac{2\pi i}{N}}}\right]=-\frac{1}{\sin\left(\frac{\pi}{N}\right)}\]

The other eigenstates of $H_-$ will be of the form $\xi_{k_1}^{\dagger}\xi_{k_2}^{\dagger} \dots \xi_{k_l}^{\dagger} |\Omega\>$ with extra energy above the ground energy of $|\Lambda_{k_1}|+|\Lambda_{k_2}| +\dots |\Lambda_{k_l}|$.
Therefore in the case where $N$ is not a multiple of 4, $\Lambda_k \neq 0$ for any $k$, and so $|\Omega\>$ is the unique ground state of $H_-$.

In order for $|\Omega\>$ to also be a true eigenstate of the full Hamiltonian $H$ with the same energy, we will also need $P|\Omega\>=-|\Omega\>$. Note that, for $N$ even $|\{k:\Lambda_k<0\}|=N/2$, so 
\[\mathcal{N}|\Omega\>=\sum_i c_i^{\dagger}c_i|\Omega\>=\sum_k \eta_k^{\dagger}\eta_k|\Omega\>=\sum_{k:\Lambda_k<0} |\Omega\>=\frac{N}{2}|\Omega\>\]
For $N$ even and not a multiple of 4, $N/2$ is odd and so  $P|\Omega\>=\exp(i\pi\mathcal{N})|\Omega\>=-|\Omega\>$ as required.

We now repeat this procedure for the $H_+=\sum c_i^{\dagger}B_{ij}c_j$ part of the Hamiltonian. 
 \[H_+=\sum c_i^{\dagger}B_{ij}c_{j} \quad \Rightarrow \quad B_{ij}=\left\lbrace \begin{array}{lc}
(\delta_{2j}-\delta_{jN})/2 &  i=1 \\
(\delta_{i+1 \: j}+\delta_{j\: i-1})/2 &  1<i<N \\
(-\delta_{1j}+\delta_{j\: N-1})/2 &  i=N \\
\end{array}
\right.\]

Let the eigenvalues and eigenvectors of $B$ be given by $\Lambda_k'$ and $\psi_{kj}$ respectively. They have a very similar form to $A$, only shifted by a factor of $\pi/N$:

\[\Lambda_k'=\cos\left(\frac{(2k+1)\pi}{N}\right) \quad \psi_{kj}=\frac{1}{\sqrt{N}}e^{i(2k+1)\pi j/N}, \quad 0\leqslant k \leqslant N-1\]
In the case where $N$ is even but not divisible by 4, the ground state of $H_+$ is fourfold degenerate and has energy

\[\sum_{k:\Lambda_k' \leqslant 0} \Lambda_k=\sum_{n=\frac{N-2}{4}}^{\frac{3N-2}{4}} \cos \left(\frac{(2n+1)\pi}{N}\right)=\Re\left[e^{i\pi(\frac{2(N-2)}{4}+1)/N}\sum_{n=0}^{\frac{N}{2}-1} e^{\frac{2\pi in}{N}}\right]\]
\[=\Re \left[i \frac{2}{1-e^{\frac{2\pi i}{N}}}\right]=-\frac{1}{\tan\left(\frac{\pi}{N}\right)}\]

Since $-\sin(\pi/N)^{-1}<-\tan(\pi/N)^{-1}$, we have now shown that $|\Omega\>$ is the unique ground state of the whole Hamiltonian $H$.

\subsubsection{Spectral gap}
In the $P=-1$ sector, all states orthogonal to the $|\Omega\>$ have an extra energy above the ground energy of at least $|\cos(\frac{\pi(N+2)}{2N})|=\sin(\frac{\pi}{n})$.

And in the $P=+1$ sector, all states have an extra energy above the ground energy of at least 
\[\frac{1}{\sin\left(\frac{\pi}{N}\right)}-\frac{1}{\tan\left(\frac{\pi}{N}\right)}=\frac{1-\cos\left(\frac{\pi}{N}\right)}{\sin\left(\frac{\pi}{N}\right)}=\frac{2\sin^2\left(\frac{\pi}{2N}\right)}{2\sin\left(\frac{\pi}{2N}\right)\cos\left(\frac{\pi}{2N}\right)}=\tan\left(\frac{\pi}{2N}\right)\]

In particular one of the four eigenstates of $H_+$ that has energy $-\tan(\pi/N)^{-1}$ will have $P=+1$, and so will be a true eigenstate of the total Hamiltonian $H$.
Therefore the spectral gap is exactly $\tan\left(\frac{\pi}{2N}\right)$.

\subsection{Spin correlation functions}
We need to calculate the quantity $\<\Omega|X_iX_j+Y_iY_j|\Omega\>$, but given the symmetry of the problem it will suffice to calculate $\rho_{ij}=\<\Omega|X_iX_j|\Omega\>$. Note that because of the translational invariance of the original Hamiltonian, we expect $\rho_{ij}$ to depend only on $|i-j|$.

For arbitrary $i<j$,
\[X_iX_j=(a_i^{\dagger}+a_i)(a_j^{\dagger}+a_j)=(c_i^{\dagger}+c_i)\exp(i\pi\sum_{k=i}^{j-1}c_k^{\dagger}c_k)(c_j^{\dagger}+c_j)\]
\[=\dots =(c_i^{\dagger}-c_i)\prod_{k=i+1}^{j-1}\left((c_k^{\dagger}+c_k)(c_k^{\dagger}-c_k)\right)(c_j^{\dagger}+c_j)\]
So\[\rho_{ij}=\<\Omega|B_iA_{i+1}B_{i+1}A_{i+2}\dots A_{j-1}B_{j-1}A_{j}|\Omega\>\]
where we have defined the operators $A_i=c_i^{\dagger}+c_i, B_i=c_i^{\dagger}-c_i$, which obey the following commutation rules:
$\{A_i,B_j\}=0,\{A_i,A_j\}=\delta_{ij}, \{B_i,B_j\}=-2\delta_{ij}$.
We can use Wick's Theorem to express this expectation in terms of a sum over all possible contractions.

We have 
\[\<\Omega|B_lA_m|\Omega\>=\<\Omega|(c_l^{\dagger}-c_l)(c_m^{\dagger}+c_m)|\Omega\>=\sum_{k,k'} \<\Omega|(\phi_{kl}^{*}\eta_k^{\dagger}-\phi_{kl}\eta_k)(\phi_{k'm}^{*}\eta_{k'}^{\dagger}+\phi_{k'm}\eta_{k'})|\Omega\>\]
\[=\sum_{k}\phi_{kl}^{*}\phi_{km}\<\Omega|\eta_k^{\dagger}\eta_k|\Omega\>-\phi_{kl}\phi_{km}^{*}\<\Omega|\eta_k\eta_k^{\dagger}|\Omega\>=\frac{1}{N}\left(\sum_{k,\Lambda_k<0}e^{2\pi ik(m-l)/N}-\sum_{k,\Lambda_k\geqslant 0}e^{2\pi ik(l-m)/N}\right)\]
Noting that this expression depends only on $l-m$, we define $G_{l-m}=\<\Omega|B_lA_m|\Omega\>$ in order to match the formulation in~\cite{Lieb1961}. Letting $r=l-m$ and pairing up terms $k$ and $N-k$, we get
\[G_r=\dots=\frac{1}{N}\left[(1-(-1)^r)+2\sum_{k=1}^{\frac{N-2}{4}} \cos(2\pi kr/N)-2\sum_{k=\frac{N+2}{4}}^{\frac{N}{2}-1}\cos(2\pi kr/N)\right]\]
Then we pair up the terms $k$ and $N/2-k$, which cancel out for even $r$, giving $G_r=0$ if $r$ is even. But for odd $r$:
\[\sum_{k=1}^{\frac{N-2}{4}}\cos(2\pi kr/N)=-\sum_{k=\frac{N+2}{4}}^{\frac{N}{2}-1}\cos(2\pi kr/N)=-\frac{1}{2}+\frac{(-1)^{\frac{r+1}{2}}}{2\sin\left(\frac{\pi r}{N}\right)}\]
So
\[G_r=\frac{2(-1)^{\frac{r+1}{2}}}{N\sin\left(\frac{\pi r}{N}\right)} \qquad r \text{ odd}\]

Then from Wick's Theorem (and the calculations $\<A_lA_m\>=\delta_{lm}=\<B_lB_m\>$ which can be found in~\cite{Lieb1961}), we get $\rho_{ij}=\det R_{n,N}$, where $n=i-j$ and $R_{n,N}$ is the $n\times n$ matrix with entries:
\begin{equation}
(R_{n,N})_{lm}=G_{l-m+1}
\label{eqn:RnN}
\end{equation}

In~\cite{Lieb1961} (and the rest of the literature) the limit $N\rightarrow \infty$ is taken earlier as this simplifies the calculation of $G_{l-m}$ (especially in the anisotropic case), and again does not affect the leading order asymptotics. This then gives a matrix $R_n$, with entries:
\begin{equation}(R_n)_{lm}= 
\quad  \left\lbrace \begin{array}{cc}
0 & l-m \text{ odd} \\
(-1)^{\frac{l-m}{2}}\frac{2}{\pi (l-m+1)}& l-m \text{ even}\\
\end{array}\right.
\label{eqn:Rn}
\end{equation}

Given the known limit $\lim_{x\rightarrow 0}x/\sin x=1$ we see that each entry in $R_n$ is the limit as $N\rightarrow \infty$ of the corresponding entry in $R_{n,N}$. Since $\det$ is a polynomial, and hence continuous, function of the entries of a matrix, for fixed $n$
\[\det R_n=\lim_{N\rightarrow \infty} \det R_{n,N}.\]
However we will need a stronger result for the behaviour as $n\rightarrow\infty$ and $N=\text{poly}(n)$. The exact result is given in Lemma \ref{lem:Requiv} below. We will first need to determine the determinant of the matrix $R_n$. This was previously obtained in~\cite{McCoy68}; for completeness, we include an alternative concise proof based on the general theory of Toeplitz matrices and a recent result of Ehrhardt~\cite{Ehrhardt2001} on the Fisher-Hartwig conjecture. 

\subsection{Toeplitz matrices}
A matrix of the following diagonal form is known as a Toeplitz matrix, see for example~\cite{Bottcher1999} for an introduction. 
\[T_n= \left(\begin{array}{ccccc}
t_0 & t_1 & t_2 &\dots & t_n \\
t_{-1}& t_0 &t_1 & \ddots &\vdots\\
t_{-2}& t_{-1} & t_0 &\ddots &\\
\vdots & \ddots &\ddots &t_0 &t_1\\
t_{-n} & \dots&& t_{-1}& t_0\\

\end{array}\right)
\]
These matrices have been especially well studied when there exists a complex-valued function $f$ integrable on the unit circle such that $\{t_k\}$ are the Fourier coefficients of $f$:
\begin{equation}
t_k=\frac{1}{2\pi}\int_{0}^{2\pi}f(e^{i\theta})e^{-ik\theta}d \theta, \quad k \in \mathbb{Z}
\label{eqn:coeff}
\end{equation}
The function $f$ is called the \emph{symbol} of the corresponding Toeplitz matrix $R_n$. We can find such a symbol by calculating the Fourier series with coefficients $t_k$. If $f(z)$ is sufficiently smooth, then the famous Szeg\"o limit theorem gives the asymptotic behaviour of $\det(R_n)$ as $n\rightarrow \infty$. However in our case,  the symbol of $R_n$ is given by the function $g(z)$, which is discontinuous:
\begin{equation} g(z)=\sum_{k=-\infty}^{\infty} t_k z^k=\left\lbrace \begin{array}{cc}
z & \arg(z)\in [0,\frac{\pi}{2})\cup (\frac{3\pi}{2},2\pi)\\
-z & \arg(z)\in (\frac{\pi}{2},\frac{3\pi}{2})\\
0 & z=\pm i\\
\end{array}\right.
\label{eqn:Toef}
\end{equation}
Note that this function $g$ is only defined on the unit circle $S=\{z:|z|=1\}$. In fact the value of $g(z)$ at $z=\pm i$ is unimportant as the coefficients $t_k$ as defined by (\ref{eqn:coeff}) will be the same for functions that agree almost everywhere (i.e. functions that differ only on a set of Lebesgue measure zero). 

The Fisher-Hartwig conjecture is an attempt to generalise the Szeg\"o limit theorem to piecewise continuous symbols, such as $g(z)$; see~\cite{Bottcher1999} and references therein. The Fisher-Hartwig conjecture has previously been used to study correlation functions of the XY model, for example in~\cite{Ovchinnikov2007}.
 A function $f$ is said to have Fisher-Hartwig singularities if it can be written in the following form: 

\begin{equation} f(z)=e^{V(z)}z^{\sum_{j=0}^{m}\beta_j}\prod_{j=0}^m |z-z_j|^{2\alpha_j} g_{z_j,\beta_j}(z)z_j^{-\beta_j} 
\label{eqn:FHsing}
\end{equation}
\[\text{where } z_j=e^{i\theta_j}, \: \theta_j \in [0,2\pi),\: \text{and } g_{z_j,\beta_j}= \left\lbrace \begin{array}{cc}
e^{i\pi \beta_j} & 0<\arg(z)<\theta_j\\
e^{-i \pi \beta_j} & \theta_j<\arg(z)<2\pi\\
\end{array}\right. \]

The Fisher-Hartwig conjecture states that for $V(z)$ sufficiently smooth and $\alpha_j,\beta_j \in \mathbb{C}$, then 
\[\det(R_n)\sim E[f] \exp(nV_0)  n^{\sum_{j=0}^{m}\beta_j} \text{ as } n\rightarrow \infty\]
where $E[f]$ is a constant independent of $n$, and $V_0=\frac{1}{2\pi}\int_{0}^{2\pi} V(e^{i\theta})d\theta$ is the zero-th Fourier coefficient of $V$.

This conjecture has now been proved to be true for a number of different conditions on $\alpha_j,\beta_j$ and $V$; but in particular Ehrhardt~\cite{Ehrhardt2001} showed that this holds under the following conditions:
\begin{description}
\item[i)] $\Re \:\alpha_j > -1/2$ for all $j$
\item[ii)] $\max_{j,k}|\Re \:\beta_j-\Re \:\beta_k|<1$ for all $j,k$
\item[iii)] $\alpha_j \pm \beta_j \neq -1,-2,-3,\dots$ for all $j$
\item[iv)] $V \in C^{\infty}$ ($V$ is infinitely differentiable).
\end{description}
An alternative proof of this result is presented in~\cite{Deift2012}.

The function $g(z)$ can be written in the form of (\ref{eqn:FHsing}) with two Fisher-Hartwig singularities:
\[g(z)=zg_{i,-\frac{1}{2}}(z)g_{-i,-\frac{1}{2}}(z) \]
Note that for $j=0,1$, $\alpha_j=0,\beta_j=-1/2$ and $V\equiv 0$, so all of the four conditions \textbf{(i)-(iv)} are satisfied. In fact the exact constant can also be calculated:
\[\det(R_n)\sim E n^{-\frac{1}{2}}\]
\[\text{where } E= \sqrt{2}G(\tfrac{1}{2})^2G(\tfrac{3}{2})^2=\frac{2^{2/3}\sqrt{e}}{A^6} \approx 0.58835...\]
where $G$ is the Barnes G-function and $A$ is the Glaisher-Kinkelin constant. This agrees with the result previously obtained in~\cite{McCoy68}.

\subsection{Relation between $R_{n,N}$ and $R_n$}
We now prove the following Lemma:

\begin{lemma}
\label{lem:Requiv}
Let $R_{n,N}$ and $R_n$ be the $n\times n$ matrices defined in equations (\ref{eqn:RnN}) and (\ref{eqn:Rn}). Then, if $N=\Omega(n^{\alpha})$ for $\alpha > 7/4$,
\[\det(R_{n,N}) \sim \det(R_n) \text{ as } n \rightarrow \infty. \]
\end{lemma}
\textbf{Remark.} We have strong numerical evidence to suggest that this result still holds under more relaxed conditions for $N$, perhaps even just $N>n$. However the result as stated here is enough for our purposes.
\begin{proof}[\textbf{\emph{Proof}}]
The proof uses the standard Weyl's inequality for Hermitian matrices to relate the eigenvalues of matrices closely related to $R_n$ and $R_{n,N}$. First introduce the permutation matrix $S_n$ with entries $(S_n)_{ij}=\delta_{i,n+1-j}$. Then $R_nS_n$ and $R_{n,N}S_n$ are both real symmetric matrices. Let $\{\lambda_i\}$ and $\{\mu_i\}$ be the ordered eigenvalues of $R_nS_n$ and $R_{n,N}S_n$ respectively. Then, by Weyl's inequality:
\[\lambda_i-\|B\| \leqslant \mu_i \leqslant \lambda_i+\|B\|\]
where $B=R_{n,N}-R_n$ is the difference between these two matrices. Then divide through by $\lambda_i$ (considering the cases $\lambda_i>0$ and $\lambda_i<0$ separately), to obtain
\[\left(1-\frac{\|B\|}{|\lambda_i|}\right)\leqslant\frac{\mu_i}{\lambda_i}\leqslant\left(1+\frac{\|B\|}{|\lambda_i|}\right).\]

From the general theory of Toeplitz matrices with bounded symbols, we know that $\|R_n\| \leqslant \|g\|_{\infty}=1$, where $g$ is as defined in (\ref{eqn:Toef}). So for all $i$, $|\lambda_i|\leqslant\|R_nS\| = \|R_n\| \leqslant 1$. Also, $\prod_i \lambda_i=\det(R_n)~E/n^{1/2}$ as $n\rightarrow \infty$, so there exists $E'$ such that for sufficiently large $n$, $|\lambda_i|>E'/n^{1/2}$. Then taking the product over $i$:

\begin{equation}
\left(1-E'^{-1}\|B\|n^{1/2}\right)^n\leqslant\prod_{i=1}^n\frac{\mu_i}{\lambda_i}=\frac{\det(R_{n,N})}{\det(R_n)}\leqslant\left(1+E'^{-1}\|B\|n^{1/2}\right)^n
\label{eqn:detineq}
\end{equation}

Note that $B$ is another Toeplitz matrix, with entries $B_{ij}=b_{i-j+1}$, where $b_r=0$ for even $r$, and for odd $r$ :
 \[b_r=(-1)^{\frac{r-1}{2}}\frac{2}{\pi r}\left(\frac{\frac{\pi r}{N}}{\sin(\frac{\pi r}{N})}-1\right) \qquad r \text{ odd}\] 
To deal with the term in brackets, we define $h(x)=x/\sin x -1$. Note that since $h'(0)=0, \: h''(0)=1/3$ and $h''(x)$ is continuous in a region around zero, there exists $C>0$ and $\delta>0$ such that for all $x \in (-\delta,\delta)$:
 \[h(x)=\frac{x}{\sin(x)}-1=\frac{x^2}{6}+O(x^4)<Cx^2.\]

Then for all $i$ and $j$, and $N\geqslant \pi n/\delta$, we can bound $|B_{ij}|\leqslant |b_n|<2C\pi n/N^2$ so we can upper bound the operator norm $\|B\|_{\infty}$, by the Hilbert-Schmidt norm:
\[\|B\|_{\infty}\leqslant \|B\|_{2}=\sqrt{\sum_{i,j=1}^n |B_{ij}|^2}\leqslant\frac{2C\pi}{N^2}\sqrt{\sum_{i,j=1}^n n^2}=\frac{2C\pi n^2}{N^2}.\] 
So if $N=\Omega(n^{\alpha})$, for some $\alpha>7/4$, then there exists a constant $D>0$ such that
\[\left(1+E'^{-1}\|B\|n^{1/2}\right)^n \leqslant \left(1+\frac{Dn^{\frac{7}{2}-2\alpha}}{n}\right)^n \leqslant \exp(Dn^{\frac{7}{2}-2\alpha})\rightarrow 1 \text{ as } n \rightarrow \infty\]
where we have used the inequality $1+x\leqslant e^{x}$.  Similarly, we can use the inequality $1-x\geqslant e^{-2x}$ which is valid for $x \in [0, \epsilon]$ for some $\epsilon>0$. Then for $n$ sufficiently large
\[\left(1-E'^{-1}\|B\|n^{1/2}\right)^n \geqslant \left(1-\frac{Dn^{\frac{7}{2}-2\alpha}}{n}\right)^n \geqslant \exp(-2Dn^{\frac{7}{2}-2\alpha})\rightarrow 1 \text{ as } n \rightarrow \infty\]

Finally, combining these results into equation (\ref{eqn:detineq}), we get
\[\lim_{n\rightarrow \infty} \frac{\det(R_{n,N}) }{ \det(R_n)}=1. \]
\end{proof}

The following lemma summarises the above discussion:

\begin{replemma}{lem:xy}
Fix $N$ even but not a multiple of 4, and let $H = \sum_{i=1}^{N-1} (X_i X_{i+1} + Y_i Y_{i+1}) + X_1 X_N +Y_1 Y_N$. Then $H$ has a nondegenerate ground state $|\Omega\>$ and spectral gap $\Omega(1/N)$. Further, for any pair $i$, $j$ such that $|i-j| = n$ and $n = o(N^{4/7})$, $\<\Omega|X_{i}X_{j}+Y_{i}Y_{j}|\Omega\> = \Omega(n^{-1/2})$. There is an efficient classical algorithm to compute the spectral gap and all the correlation functions.
\end{replemma}

\section{Outlook}
\label{sec:outlook}

Although we have translated some previous results about QMA-completeness of the \lham\ problem closer to truly physically realistic systems, one significant issue remaining is the large weights required for the interactions in the hard instances. This requirement is a basic limitation imposed by our use of perturbation theory for gadgets. However, it would be very interesting if recently developed techniques proving QMA-hardness with lower-weight interactions~\cite{cao14,childs14,Childs2015} could be extended or combined with our results.

There are also some specific open questions left over. First, the natural cases of the general Heisenberg model on a square lattice, and the antiferromagnetic Heisenberg model on a triangular lattice, are still unresolved. Proving these cases QMA-complete could follow from finding an exactly solvable (but nontrivial) special case of the Heisenberg model whose interactions are suitably sparse. Second, there is scope for tightening our classification of interactions of the form $\alpha XX + \beta YY + \gamma ZZ$ with fixed signs. For many of the interactions of this form, we know the interaction is contained within StoqMA, but not whether it is StoqMA-complete, or within some smaller complexity class. A particularly natural example which is currently unknown is the {\em ferromagnetic} XY model, whose interactions are of the form $-XX-YY$. Third, we have not been able to resolve the complexity of every set of 2-qubit interactions with fixed signs. It remains to classify interactions with nonzero 1-local part, and the case of the \spham\ problem where $|\mathcal{S}| \geqslant 2$.

These open questions, along with our results described here, highlight that the \lham\ problem displays a rich and complex structure when restrictions on interaction signs and topologies are considered.

\subsection*{Acknowledgements}

SP was supported by the UK EPSRC. AM was supported by the UK EPSRC under Early Career Fellowship EP/L021005/1 and would like to thank Toby Cubitt for enlightening discussions on the topic of this paper. We would also like to thank Jacques Perk for pointing out references.

\bibliographystyle{plain}
\bibliography{MyLibrary}

\end{document}